\documentclass[a4paper]{article}

\usepackage{ifthen}
\usepackage{comment}
\makeatletter
\ifthenelse{\not{\isundefined{\lipics}}}
{
\makeatletter
\let\c@author\relax
\makeatother
}{}
\usepackage[
style=numeric,
giveninits=true, 
maxbibnames=99, 
maxcitenames=3, 
natbib=true, 
labelnumber, 
backend=biber, 
sortcites=true,
doi=false,isbn=false,url=false
]{biblatex}
\renewbibmacro{in:}{%
  \ifentrytype{article}{}{\printtext{\bibstring{in}\intitlepunct}}
}
\let\oldcitet\citet
\renewcommand{\citet}[1]{\mbox{\oldcitet{#1}}}

\usepackage[margin=1in]{geometry}

\usepackage[shortlabels]{enumitem}

\ifthenelse{\isundefined{\llncs}}{
  \usepackage{amsthm}
}

\usepackage{amsmath,amsfonts,amssymb}
\usepackage{mathtools}
\usepackage{mathrsfs}
\usepackage{thm-restate}
\usepackage{wrapfig}
\usepackage{xspace}
\usepackage{csquotes}
\usepackage{multirow}

\usepackage[bibliography=common]{apxproof}

\ifthenelse{\isundefined{\lipics}}
{

\usepackage[colorlinks=true,bookmarks=false]{hyperref}
\usepackage[nameinlink,capitalise]{cleveref}
\usepackage[small,bf]{caption}
\usepackage[pdftex]{graphicx}
\usepackage{subcaption}
\usepackage[british]{babel}
}
{}

\usepackage{xcolor}
\definecolor{darkblue}{rgb}{0,0,0.45}
\definecolor{darkred}{rgb}{0.6,0,0}
\definecolor{darkgreen}{rgb}{0.13,0.5,0}
\hypersetup{colorlinks, linkcolor=darkblue, citecolor=darkgreen,
urlcolor=darkblue}

\usepackage{aliascnt}
\theoremstyle{plain}
\newtheorem{theorem}{Theorem}%
\crefname{theorem}{Theorem}{Theorems}
\crefname{observation}{Observation}{Observations}
\newtheorem{lemma}[theorem]{Lemma}%
\newtheorem{claim}[theorem]{Claim}%
\newtheorem{definition}[theorem]{Definition}

\newcommand{\mc}[1]{{\mathcal{#1}}}

\newcommand{\eps}{{\varepsilon}}

\newcommand{\hy}{\hbox{-}\nobreak\hskip0pt}

\usepackage{environ,xstring}
\newif\iflabel
\newif\ifdbs
\newif\ifamp
\NewEnviron{doitall}{%
  \noexpandarg
  \expandafter\IfSubStr\expandafter{\BODY}{\label}{\labeltrue}{\labelfalse}%
  \expandafter\IfSubStr\expandafter{\BODY}{\\}{\dbstrue}{\dbsfalse}%
  \expandafter\IfSubStr\expandafter{\BODY}{&}{\amptrue}{\ampfalse}%
  \iflabel\def\doitallstar{}\else\def\doitallstar{*}\fi
  \ifdbs
    \ifamp
      \def\doitallname{align}%
    \else
      \def\doitallname{multline}%
    \fi
  \else
    \def\doitallname{equation}
  \fi
  \begingroup\edef\x{\endgroup
    \noexpand\begin{\doitallname\doitallstar}%
    \noexpand\BODY
    \noexpand\end{\doitallname\doitallstar}%
  }\x
}
\def\[#1\]{\begin{doitall}#1\end{doitall}}

\newcommand{\pname}[1]{\textsc{#1}}

\newcommand{\remove}[1]{}

\usepackage{framed}
\newcommand{\executeiffilenewer}[3]{%
\ifnum\pdfstrcmp{\pdffilemoddate{#1}}%
{\pdffilemoddate{#2}}>0%
{\immediate\write18{#3}}\fi%
}

\makeatletter
\newcommand{\customlabel}[2]{%
   \protected@write \@auxout {}{\string \newlabel {#1}{{#2}{\thepage}{#2}{#1}{}} }%
   \hypertarget{#1}{}
}
\makeatother

\bibliography{papers}

\DeclareMathOperator{\cost}{cost}
\DeclareMathOperator{\dist}{dist}
\DeclareMathOperator{\diam}{diam}
\newcommand{\SF}{\pname{Steiner Forest}\xspace}
\newcommand{\ST}{\pname{Steiner Tree}\xspace}
\newcommand{\pow}[1]{\lfloor #1\rfloor_{2}}
\title{Parameterized Algorithms for\\
Steiner Forest in Bounded Width Graphs}

\author{Andreas Emil Feldmann\thanks{Department of Computer Science, University 
of Sheffield, UK}
\and
Michael Lampis\thanks{Universit\'{e} Paris-Dauphine, PSL University, CNRS 
UMR7243, LAMSADE, Paris, France, partially supported by ANR project 
ANR-21-CE48-0022 (S-EX-AP-PE-AL)}
}

\date{}

\begin{document}

\maketitle

\begin{abstract}

In this paper we reassess the parameterized complexity and approximability of
the well-studied \SF problem in several graph classes of bounded width. The
problem takes an edge-weighted graph and pairs of vertices as input, and the
aim is to find a minimum cost subgraph in which each given vertex pair lies in
the same connected component.  It is known that this problem is APX-hard in
general, and NP-hard on graphs of treewidth~3, treedepth~4, and feedback vertex
set size~2. However, Bateni, Hajiaghayi and Marx~[JACM,~2011] gave an approximation scheme
with a runtime of~$n^{O(k^2/\eps)}$ on graphs of treewidth~$k$. Our main result
is a much faster \emph{efficient parameterized approximation scheme (EPAS)}
with a runtime of~$2^{O(\frac{k^2}{\eps}\log\frac{k}{\eps})}\cdot n^{O(1)}$. If
$k$ instead is the vertex cover number of the input graph, we show how to
compute the optimum solution in~$2^{O(k\log k)}\cdot n^{O(1)}$ time, and we
also prove that this runtime dependence on $k$ is asymptotically best possible,
under ETH.  Furthermore, if~$k$ is the size of a feedback edge set, then we
obtain a faster~$2^{O(k)}\cdot n^{O(1)}$ time algorithm, which again cannot be
improved under~ETH.

\end{abstract}

\section{Introduction}

The \SF problem is one of the most well-studied problems in network
design~\cite{ljubic2021solving, hwang1992steiner, du2013advances,
gupta2011approximation}. In this problem the input consists of a graph
$G=(V,E)$ with positive edge weights, %
a set of \emph{terminals} $R\subseteq V$, and a set of
\emph{demands}~$D\subseteq\binom{R}{2}$. The objective is to select a subgraph
$F\subseteq G$, minimizing the total cost of selected edges, while ensuring that for every
demand pair $\{s,t\}\in D$, $s$ and $t$ are in the same connected component of
$F$.  Since edge weights are positive, it is easy to see that the optimal
solution is always a forest. The \SF problem finds many applications (see
surveys~\cite{ljubic2021solving, cheng2013steiner, voss2006steiner,
tang2020survey}), for example in telecommunication networks
(cf.~\cite{voss2006steiner}). 

Our goal in this paper is to reassess the complexity of this fundamental problem
from the point of view of parameterized complexity and approximation
algorithms.\footnote{We assume the reader is familiar with the
basics of parameterized complexity and approximation algorithms, such as the classes FPT and APX
and the definition of treewidth, as given in standard textbooks 
\cite{cygan2015parameterized,williamson2011design,feldmann2020survey}. 
We give full definitions of all parameters in \cref{sec:prelims}.}
In order to recall the context, it is helpful to compare \SF to the
even more well-studied \ST problem, which is the special case of \SF where all
terminals are required to be connected, i.e., $D=\binom{R}{2}$, and an optimal
solution is a tree. \ST was already included in Karp's seminal
list~\cite{karp1975computational} of NP-hard problems from the 1970s. From the
approximation point of view, \ST (and therefore \SF) is known to be
APX-hard~\cite{chlebik2008steiner}, but both problems admit constant factor
approximations in polynomial time for general input graphs, where the best
approximation factors known are~$\ln(4)+\eps<1.39$~\cite{byrka2013steiner}
and~$2$~\cite{agrawal1991trees, ravi1994primal}, respectively. Despite this
similarity, when considering graph width parameters the problems exhibit 
wildly divergent behaviors from the parameterized complexity point of
view: whereas \ST is FPT parameterized by standard structural parameters such
as treewidth and can in fact even be solved in single exponential
$2^{O(k)}n^{O(1)}$ time~\cite{bodlaender2015deterministic} when $k$ is the
treewidth, \SF is NP-hard on graphs of treewidth~$3$, as shown independently by
\citet{gassner2010steiner} and \citet{bateni2011approximation}.

\SF is therefore a problem that presents a dramatic jump in complexity in this
context, compared to \ST, as the hardness result on graphs of treewidth~$3$
rules out even an XP algorithm for parameter treewidth. One of the main
positive contributions of \citet{bateni2011approximation} was an algorithm
attempting to bridge this gap using approximation. In particular, they showed
that \SF admits an approximation scheme for graphs of treewidth~$k$, which
computes a $(1+\eps)$\hy{}approximation in $n^{O(k^2/\eps)}$ time for any $\eps>0$. Hence, if we
allow slightly sub-optimal solutions, we can at least place the problem in XP
parameterized by treewidth. In their paper, \citet{bateni2011approximation}
remark that because the exponent of the polynomial of this runtime depends on
$k$ and $\eps$, "it remains an interesting question for future research whether
this dependence can be removed", that is, whether a $(1+\eps)$-approximation
can be obtained in FPT time.

The main result of our paper is a positive resolution of the question of
\citet{bateni2011approximation}: we show that \SF\ admits an \emph{efficient
parameterized approximation scheme (EPAS)} for treewidth, that is, a
$(1+\eps)$-approximation algorithm with a runtime of the form
$f(k,\eps)n^{O(1)}$. In other words, we show that their algorithm can be
improved in a way that makes the running time FPT not only in the treewidth,
but also in $1/\eps$. More precisely, we show the following:

\begin{theorem}\label{thm:tw-EPAS}
The \SF problem admits an EPAS parameterized by the treewidth $k$ with a runtime of $2^{O(\frac{k^2}{\eps}\log\frac{k}{\eps})}\cdot n^{O(1)}$.
\end{theorem}

Moving on from treewidth, we ask what the most general
parameter is for which we may hope to obtain an FPT \emph{exact} algorithm for
\SF. We observe that the NP-hardness result of \citet{gassner2010steiner,
bateni2011approximation} for \SF on graphs of treewidth~$3$ actually has some
further implications for some even more restricted parameters: the graphs
constructed in their reductions also have constant \emph{treedepth} and
\emph{feedback vertex set} size, implying that the problem remains hard for
both of these parameters (which are incomparable in general). More
precisely, known reductions imply the following:

\begin{theorem}[\citet{gassner2010steiner,bateni2011approximation}]
The \SF problem is NP-hard on graphs of treewidth~$3$, treedepth~$4$, and feedback vertex set of size~$2$.
\end{theorem}

This leads us to consider even more restricted parameters, such as the size of
a \emph{vertex cover} and \emph{feedback edge set}, which are not bounded in
this reduction. Indeed, not only do we prove that \SF is FPT for both of these
parameters, but we are also able to determine the correct parameter dependence,
under the Exponential Time Hypothesis (ETH). For feedback edge set the optimal dependence is single
exponential:  

\begin{theorem}\label{thm:fes}
The \SF problem is FPT parameterized by the size $k$ of a feedback edge set and can be solved in $2^{O(k)}n^{O(1)}$ time.
Furthermore, no $2^{o(k)}n^{O(1)}$ time algorithm exists, under ETH.
\end{theorem}

For the parameterization by the vertex cover size, we obtain a slower runtime
for our FPT algorithm. Interestingly, we are also able to prove that this is
best possible, under~ETH. Our lower bound for \SF is in contrast to the \ST
problem, for which a faster $2^{O(k)}n^{O(1)}$ time algorithm exists, even if
$k$ is the treewidth~\cite{bodlaender2015deterministic}.

\begin{theorem}\label{thm:vc}
The \SF problem is FPT parameterized by the size $k$ of a vertex cover and can be solved in $2^{O(k\log k)}n^{O(1)}$ time. Furthermore, no $2^{o(k\log k)}n^{O(1)}$ time algorithm exists, under ETH.
\end{theorem}

We remark that \citet{bodlaender2023complexity} recently independently showed 
that \SF admits a~$2^{O(k\log k)}n^{O(1)}$ time algorithm for the size $k$ of a 
vertex cover (improving an algorithm for the unweighted version of the problem 
given in~\cite{GimaHKKO22}). While they develop their own dynamic program to 
solve this problem, we rely on an existing algorithm 
by~\cite{bateni2011approximation} (see \autoref{thm:tw-dp}). Accordingly our 
description of the algorithm is very short compared 
to~\cite{bodlaender2023complexity}. The more interesting part of 
\autoref{thm:vc} however is the proof of the lower bound.

\subsection{Overview of Techniques}

Let us briefly sketch the high level ideas of our results given by \cref{thm:tw-EPAS,thm:vc,thm:fes}.

\smallskip
\textbf{EPAS for treewidth.} Our algorithm extends
the work of \cite{bateni2011approximation}, so let us briefly recall some key
ideas. 
Given a rooted tree decomposition, a terminal $t$ is called \emph{active} for a bag $B$
if there is a demand $\{s,t\}\in D$ such that $t$ lies in the sub-tree rooted at $B$
while $s$ does not (see \cref{sec:prelims} for formal definitions). 
It is a standard property of tree decompositions that every bag is a separator.
Hence the
component of any feasible solution that contains an active terminal must intersect~$B$. 
The hardness of the problem now inherently stems from the fact
that we have to decide for all active terminals of a bag, how
the corresponding component intersects the bag, and therefore how the active
terminals (whose number is unbounded by $k$) are partitioned into connected
components.  Suppose, however, that someone supplied us with this information,
that is, suppose that for each bag $B$ we are given a set of partitions $\Pi_B$
of its active terminals and we are promised that the optimal solution
\emph{conforms} to all $\Pi_B$. By this we mean that if we look at how the
optimal solution partitions the active terminals of $B$ into connected
components and call this partition $\pi$, then $\pi\in \Pi_B$, that is, the
optimal partition is always one of the supplied options. In this case, using
this extra information, the problem does become tractable, as shown in~\cite{bateni2011approximation}:

\begin{theorem}[\citet{bateni2011approximation}]\label{thm:tw-dp} For an input
graph $G$ on $n$ vertices, let a rooted nice tree decomposition of width~$k$ be
given, such that all terminals lie in bags of leaf nodes of the decomposition.
Also, let a set~$\Pi_B$ of partitions of the active terminals of each bag $B$ of
the decomposition be given. If $p=\sum_B |\Pi_B|$ is the total number of
partitions, then a minimum cost \SF solution conforming to all $\Pi_B$ can be
computed in $2^{O(k\log k)}\cdot(pn)^{O(1)}$ time.  
\end{theorem}

The above theorem does not seem immediately helpful, since one would still need to
find a small collection of partitions $\Pi_B$ in order to obtain an efficient algorithm. 
Note however, that the partition sets may conform to an approximate solution as well, which
would let the algorithm compute a solution that it at least as good.
The strategy of \cite{bateni2011approximation} therefore is to
construct a collection of partitions that has size polynomial in~$n$ (when $k,\eps$ are
fixed constants) by stipulating that when two active terminals are ``close'' to
each other, they should belong in the same set of the partition of some near-optimal solution. In order to
bound the resulting approximation ratio, they need to provide a
charging scheme: starting from an optimal solution, they merge components which are
``close'', to obtain a solution that conforms to the~$\Pi_B$ used by the algorithm. 
They then show that the resulting solution is still near-optimal by charging the extra
cost incurred by a merging operation to one of the two merged components.

A blocking point in the above is that we need to make sure we do not
``overcharge'' any component. This is accomplished in~\cite{bateni2011approximation}
via a partial ordering of the components: we
order the components according to the highest bag of the rooted tree decomposition they
intersect, and whenever two components are merged we charge this to the
\emph{lower} component. As shown in~\cite{bateni2011approximation},
this ensures that no component is charged for more than
$k$ merges. Unfortunately, this also implies that the merging procedure is not
symmetric, which severely diminishes the contexts in which we can apply it.

Let us now describe how our approach improves upon this algorithm. A key
ingredient  will be a more sophisticated charging scheme, which will allow us
to obtain a better (smaller) collection of partitions~$\Pi_B$, without
sacrificing solution quality.  Counter-intuitively, we will achieve this by
introducing a \emph{second} parameter: the height $h$ of the tree
decomposition.  Informally, we will now construct a near-optimal solution by
merging two components whenever the connection cost is low compared to the cost of
(a~part~of) \emph{either} component (as opposed to the lower component). As 
in~\cite{bateni2011approximation}, this runs the risk of charging many merging
operations to a higher component, but by performing an accounting by tree
decomposition level and using the fact that the decomposition only has $h$
levels, we are able to show that our solution is still near-optimal even though
we merge components much more aggressively than~\cite{bateni2011approximation}.
In this way, for each bag we construct one partition of its active terminals into
a number of sets that is polynomial in $k+h+\frac{1}{\eps}+\log n$, in a way
that guarantees that this partition is a \emph{refinement} of a near-optimal
solution. That is, whenever we decide to place two terminals together in our
partition, the near-optimal solution does the same. However, this solution does
not necessarily conform to the resulting partitions, as two terminals of the same component
might end up in different sets of the partition for a bag.

At this point an astute reader may be wondering that since we consider both the
width $k$ and the height $h$ of the decomposition as parameters, we are
effectively parameterizing by treedepth, rather than treewidth. This is
correct, but we then go on to invoke a result of \citet{BodlaenderH98} which
states that any tree decomposition can be rebalanced to have height $O(\log n)$
without severely increasing its width. Hence, the family of partitions we now
have has size polynomial in $k+\frac{1}{\eps}+\log n$. 
However, we are not done yet, since at this point we can
only guarantee that our partitions are refinements of a near-optimal
solution. To complete the algorithm, we work from this family of partitions to
obtain a collection of partitions \emph{conforming} to our near-optimal
solutions using $\delta$-nets (this is similar to the approach of
\cite{bateni2011approximation}). This leads to a running time of the form
$(\log n)^{O(\frac{k^2}{\eps})}n^{O(1)}$, which by standard arguments of parameterized
complexity is in fact FPT and can be upper-bounded by a function of the form
$2^{O(\frac{k^2}{\eps}\log\frac{k}{\eps})}\cdot n^{O(1)}$. 
To summarize, our high-level strategy is to
show that the approach of \cite{bateni2011approximation} can be significantly
improved when the input decomposition has small width and height, but then we
observe that our new scheme is efficient enough in the height that even if we
replace $h$ by a bound that can be obtained for \emph{any} graph, we still have
an algorithm with an FPT running time, that is, significantly faster than that
of \cite{bateni2011approximation}. 

\smallskip
\textbf{Vertex Cover.} For the parameterization by the vertex cover size, as
mentioned we obtain an FPT exact algorithm with dependence $2^{O(k\log k)}$. A
similar algorithm was recently independently obtained by
\citet{bodlaender2023complexity} via dynamic programming. However, our
algorithm is significantly simpler, because our strategy is to show how to
construct a tree decomposition and a collection of partitions $\Pi_B$ such that
we only need \emph{one} partition of the active terminals for each bag. As a
consequence, $p=O(n)$ and \autoref{thm:tw-dp} implies the algorithm of
\autoref{thm:vc}, without the need to formulate a new dynamic program.

Our main result for this parameter is that under ETH the runtime dependence
is asymptotically optimal. Note that this also implies that the runtime of the dynamic program
given by \autoref{thm:tw-dp} cannot be improved with regards to the dependence
on the treewidth. To show this, we present a reduction from \textsc{3-SAT}, where the
goal is to compress an $n$-variable formula into a \SF
instance such that the graph has vertex cover size $O(n/\log n)$. The intuition on why
it is possible to achieve such a compression is the following: suppose we have
an instance with vertex cover of size $k$ and a demand between two vertices of
the independent set. Then the simplest way to satisfy such a demand is to
connect both vertices to a common neighbor in the vertex cover. This encodes a
choice among $k$ vertices, and hence it is sufficient to encode the assignment for
$\log(k)$ binary variables. The strategy of our reduction is to set up some
choice gadgets which allow us to encode the assignments to the original formula
taking advantage of the fact that each choice can represent a logarithmic
number of variables. Hence we can obtain a construction of slightly sub-linear
($O(n/\log n)$) size.  We then of course need to add some verification gadgets,
representing the clauses, to check that the formula is indeed satisfied. But
even though the number of such gadgets is linear in $n$, we make sure that they
form an independent set, and hence the total vertex cover size remains
sufficiently small to obtain our lower bound. We note that this compression
strategy is similar to techniques recently used to obtain slightly
super-exponential lower bounds for vertex cover for other problems
\cite{LampisMV23,LampisV23}, but the constructions we use are new and tailored to
\SF.

\smallskip
\textbf{Feedback Edge Set.} For the parameterization by the size $k$ of a
feedback edge set, instead of relying on the dynamic program given by
\autoref{thm:tw-dp} we go an entirely different route in order to obtain the
faster $2^{O(k)}n^{O(1)}$ time FPT algorithm of \autoref{thm:fes}. First off,
it is not hard to reduce the \SF problem to an instance in which all vertices
have degree at least~$2$.  We then consider paths with internal vertices of
degree $2$, with endpoints that are vertices incident to the feedback edge set
or vertices of degree at least~$3$.  We call these paths \emph{topo-edges} and
argue that there are only $O(k)$ of these. We then guess for which topo-edges
the two endpoints lie in different components of the optimal \SF solution,
which can be done in $2^{O(k)}$ time.  If a topo-edge has both its endpoints in
the same component of the optimum, we show that it can be easily handled. For
the remaining topo-edges, we can decide which edges along the path do not
belong to the optimal solution by a reduction to the polynomial-time solvable
\pname{Min Cut} problem.

\subsection{Related work}
\label{sec:related}

\citet{bateni2011approximation} show that one of the consequences of their XP
approximation scheme is a PTAS for \SF on planar graphs, by using the common
technique pioneered by \citet{baker1994approximation} of reducing this problem
to graphs for which the treewidth is bounded as a function of~$\eps$.  Because
their algorithm is not FPT, their PTAS has a running time of the form
$n^{f(\eps)}$. By using our algorithm from \autoref{thm:tw-EPAS} we can improve
this runtime to~$f(\eps)n^{O(1)}$, i.e., we obtain on EPTAS for planar graphs.
However, \citet{eisenstat2012efficient} already showed that a
$(1+\eps)$-approximation algorithm with a runtime of $O(f(\eps)\cdot n\log^3
n)$ exists for \SF on planar graphs. While they build on the work of
\citet{bateni2011approximation}, and in particular also reduce to graphs of
treewidth bounded as a function of~$\eps$, interestingly they do not obtain an
EPAS parameterized by treewidth.  Instead they use a different route and show
that given a graph~$H$ of treewidth~$k$, in $O(f(k,\eps)\cdot n\log^2 n)$ time
it is possible to compute a \SF solution in $H$ whose cost is at most
$\cost(F^\star)+\eps\cost(H)$, i.e., there is an additive error that depends on
the cost of~$H$ compared to the optimum solution~$F^\star$. If the input
graph~$G$ is planar, then a result by \citet{borradaile2009n} implies that from
$G$ a so-called \emph{banyan}~\cite{rao1998approximating, bartal2021near} can
be computed, which is a subgraph of $G$ with cost bounded by
$O(g(\eps)\cost(F^\star))$, and which contains a near-optimal approximation of
every Steiner forest (cf.~\cite[Lemma~2.1]{eisenstat2012efficient}). By
applying the framework of \citet{bateni2011approximation} on the banyan instead
of the input graph, it is then possible to obtain a graph $H$ of treewidth
bounded by a function of~$\eps$, for which the algorithm of
\citet{eisenstat2012efficient} computes a $(1+O(\eps))$-approximation for the
input. 

If it would be possible to compute a banyan for bounded treewidth graphs, then the algorithm of \citet{eisenstat2012efficient} would also imply an EPAS for treewidth. However, to the best of our knowledge, and as explicitly stated by \citet{bartal2021near}, banyans are only known for planar graphs~\cite{eisenstat2012efficient, borradaile2009n}, Euclidean metrics~\cite{rao1998approximating}, and doubling metrics~\cite{bartal2021near} (in fact, the latter are so-called \emph{forest banyans}, which have weaker properties). Thus it is unclear how to obtain an EPAS for \SF parameterized by the treewidth via the algorithm of \citet{eisenstat2012efficient}. We leave open whether a banyan exists for bounded treewidth graphs, which could give an alternative algorithm to the one given in \autoref{thm:tw-EPAS}. However, a further remark is that the cost of the banyan for planar graphs obtained by \citet{borradaile2009n} has exponential dependence on~$1/\eps$, which implies a double exponential runtime dependence on~$1/\eps$ for the EPTAS for planar graphs. If a banyan can be obtained for bounded treewidth graphs by generalizing the techniques of \citet{borradaile2009n} to minor-free graphs, then the resulting EPAS parameterized by treewidth would also have double exponential runtime in~$1/\eps$. In this case however, our EPAS given by \autoref{thm:tw-EPAS} would be exponentially faster.

A different parameter that is often studied in the context of Steiner problems 
is the number $p=|R|$ of terminals. The classic result of 
\citet{dreyfus1971steiner} presents an FPT algorithm for \ST with a runtime of 
$3^p n^{O(1)}$. For unweighted graphs, this was 
improved~\cite{bjorklund2007fourier, nederlof2009fast} to $2^p n^{O(1)}$, while 
the fastest known algorithm for weighted graphs can compute the optimum in 
$(2+\eps)^p n^{O(\sqrt{\frac{1}{\eps}}\log\frac{1}{\eps})}$ 
time~\cite{fuchs2007dynamic} for any~$\eps>0$. The algorithm of 
\citet{dreyfus1971steiner} can be generalized to solve \SF in $2^{O(p)} 
n^{O(1)}$ time (cf.~\cite{chitnis2021parameterized}). A somewhat dual parameter 
to the number of terminals is the number~$q$ of non-terminals (so-called 
\emph{Steiner vertices}) in the optimum solution. For this parameter, a folklore 
result states that \ST (and thus also \SF) is 
W[2]-hard~(cf.~\cite{cygan2015parameterized, dvorak2021parameterized}). However, 
an EPAS with a runtime of $2^{O(q^2/\eps^4)}n^{O(1)}$ was shown to exist for 
\ST~\cite{dvorak2021parameterized}. For \SF it is not hard to see that such an 
EPAS parameterized by $q$ cannot exist unless P=NP 
(cf.~\cite{dvorak2021parameterized}), but if $c$ denotes the number of 
components of the optimum solution, there is an EPAS with a runtime of 
$(2c)^{O((q+c)^2/\eps^4)}n^{O(1)}$~\cite{dvorak2021parameterized}. Similar 
results have been found for related Steiner problems in directed 
graphs~\cite{chitnis2021parameterized}. For further results in the area of 
parameterized approximations, we refer to the survey~\cite{feldmann2020survey}.

\section{Preliminaries}\label{sec:prelims}

As mentioned, we assume the reader is familiar with the basics of parameterized
complexity, such as the class FPT~\cite{cygan2015parameterized}, and approximation algorithms
such as a PTAS~\cite{williamson2011design}. A
\emph{parameterized approximation scheme (PAS)} is an algorithm that computes a
$(1+\eps)$-approximation for a problem in  $f(k,\eps)n^{g(\eps)}$ time for some
functions $f$ and $g$, while an an \emph{efficient parameterized approximation
scheme (EPAS)} is a $(1+\eps)$-approximation algorithm running in time
$f(k,\eps)n^{O(1)}$ (that is, the running time is FPT in $k+\frac{1}{\eps}$).
The distinction between a PAS and an EPAS is similar to the one between a PTAS
and an EPTAS.

By $w:E\to\mathbb{R}^+$ we denote an edge-weight function, so that the cost of a solution $F$
to the \SF problem is $\cost(F)=\sum_{e\in E(F)} w(e)$. 
We will use $F^\star$ to denote an optimal solution, and for
$\alpha\ge 1$ we will say that a solution $F$ is $\alpha$-approximate if
$\cost(F)\le \alpha\cost(F^\star)$. For $u,v\in V$ we use $\dist(u,v)$ to
denote the shortest-path distance from $u$ to $v$ in $G$ according to the weight function~$w$.

\begin{definition} Given a graph $G=(V,E)$, a \emph{tree decomposition} is a
pair $(T,\{B_i\}_{i\in V(T)})$, where $T$ is a tree and each node $i\in V(T)$
of the tree is associated with a \emph{bag} $B_i\subseteq V$, with the
following properties: 
\begin{enumerate}
\item $\bigcup_{i\in V(T)} B_i=V$,
i.e., all vertices of $G$ are covered by the bags, 
\item for every edge $uv\in
E$ of $G$ there exists a node $i\in V(T)$ of the tree for which $u,v\in B_i$,
and 
\item for every vertex $v\in V$ of $G$ the nodes $\{i\in V(T)\mid v\in
B_i\}$ of the tree for which the bags contain $v$ induce a (connected) subtree
of $T$.  
\end{enumerate} 
The \emph{width} of the tree decomposition is
$\max_{i\in V(T)}\{|B_i|-1\}$ and the \emph{treewidth} of $G$ is the minimum
width over all its tree decompositions.  

A rooted tree decomposition is \emph{nice} if for every $i\in V(T)$ we have one of the
following: 

\begin{enumerate} 

\item $i$ has no children\footnote{here we do not demand the leaf nodes to be
empty, as is often assumed for this definition.} ($i$ is a \emph{leaf node}),

\item $i$ has exactly two children $i_1$ and $i_2$ such that
$B_i=B_{i_1}=B_{i_2}$ ($i$ is a \emph{join node}), 

\item $i$ has a single child $i'$ where $B_i=B_{i'}\cup \{v\}$ for some $v\in
V$ ($i$ is an \emph{introduce node}), or 

\item $i$ has a single child $i'$ where $B_i=B_{i'}\setminus \{v\}$ for some
$v\in V$ ($i$ is a \emph{forget node}).  

\end{enumerate}

\end{definition}

Given a rooted tree decomposition $T$ of a graph $G$, for a node $u$ of $T$ let
$B$ be the bag associated with it.  Then $V_B$ is the set of vertices of all
bags in the subtree rooted at $u$.  The set $A_B\subseteq R$ denotes the
\emph{active terminals} of the bag $B$: for any demand pair~$\{s,t\}\in D$, if
$s\in V_B$ and $t\notin V_B$ then $s\in A_B$.  For any \SF solution $F$, if a
connected component $C$ of $F$ contains an active terminal, then we say that
$C$ is an \emph{active component} for $B$. For a fixed solution $F$, we denote
the set of all active components for $B$ by~$\mc{C}_B$.  Note that every active
component must intersect the bag $B$. 

If for every bag $B$ a set of partitions $\Pi_B$ of $A_B$ is given, a
\SF solution $F$ is \emph{conforming} to all $\Pi_B$, if for each bag~$B$ there exists a partition
$\pi\in\Pi_B$ such that any two active terminals in $A_B$ are in the same set
$S\in\pi$ if and only if they are part of the same active component $C$ of~$F$,
i.e., $S\subseteq V(C)$ and $S'\cap V(C)=\emptyset$ for any $S'\in\pi$ with
$S'\neq S$ (note that this implies $|\pi|\leq |B|$). One technicality of \cref{thm:tw-dp} is that
the algorithm needs a nice tree decomposition as input, for which the terminals
only appear in bags that are leaf nodes of the decomposition. Given any tree
decomposition, these conditions are not hard to meet
(cf.~\cite[Lemma~6]{bateni2011approximation}).  However, for our algorithms we
are going to rely on tree decompositions with certain additional properties.
Hence we will need to revisit the conditions needed for the algorithm of
\autoref{thm:tw-dp} when using it for our purposes.

We will also consider the following parameters: 
The \emph{treedepth} of a graph $G$ can be defined recursively as follows: (i) the treedepth of $K_1$ is $1$ (ii) the treedepth of a disconnected graph is the maximum of the treedepth of any of its components (iii) the treedepth of a connected graph $G$ is $1+\min_{v\in V(G)} \textrm{td}(G-v)$. 
A \emph{feedback vertex set} is a set of vertices whose removal leaves a forest. A \emph{vertex cover} is a set of vertices such that its removal leaves an edge-less graph. A \emph{feedback edge set} is a set of edges whose removal leaves a forest. In a connected graph with $n$ vertices and $m$ edges the minimum feedback edge set always has size $m-n+1$.

As part of our approximation algorithm we will use the notion of $\delta$-nets, defined as follows. 
A well-known fact is that a $\delta$-net exists for any metric and any $\delta\geq 0$, 
and it can be constructed greedily in polynomial time. 

\begin{definition}\label{dfn:net}
Given a metric $(X,\dist)$, a \emph{$\delta$-net} is a subset $N\subseteq X$ of points, such that
\begin{enumerate}
    \item any two net points $u,v\in N$ are far from one another, i.e., $\dist(u,v)>\delta$, and
    \item for any node $u\in X$ there is some net point $v\in N$ close by, i.e., $\dist(u,v)\leq\delta$.
\end{enumerate}
\end{definition}

\section{An efficient parameterized approximation scheme for treewidth}\label{sec:tw}

In this section we describe the main result of this paper which is an EPAS for
\SF parameterized by treewidth. We begin by giving two preliminary tools
(\autoref{lem:nice-tree-decomp} and \autoref{lem:aspect-ratio}) which
facilitate the algorithm by ensuring that the given tree decomposition has
logarithmic height and that the instance has aspect ratio (ratio of the weights
of the heaviest over the lightest edge) bounded by a polynomial in $n$.

We then go on to \autoref{sec:tw-height} where we introduce a second parameter,
the height $h$ of the decomposition. Our goal is to fix an almost-optimal
solution $F_\eps$ and describe an algorithm that produces a partition $\zeta_B$ of the active terminals
for each bag $B$ of the decomposition, where $\zeta_B$ is
\emph{a refinement} of the partition implied by $F_\eps$
(\autoref{lem:near-opt-1}). In other words, we seek a partition $\zeta_B$ of
$A_B$ such that if two terminals $t_1,t_2$ are in the same set of $\zeta_B$,
then they are also in the same component of $F_\eps$. Of course, it is trivial
to achieve this by giving a $\zeta_B$ where each active terminal is in its own
set, so the interesting part here is how we group terminals together in a way
that in the end allows us to bound $|\zeta_B|$ by a polynomial of
$k+h+\frac{1}{\eps}+\log n$, while still ensuring that $F_\eps$ is almost
optimal. 

The partition $\zeta_B$ of \autoref{lem:near-opt-1} is not yet conforming,
because two terminals which are in distinct sets of $\zeta_B$ may still be in
the same component of $F_\eps$, and thus we cannot apply \cref{thm:tw-dp} at this point. 
Therefore in \autoref{sec:tw-log-height}, given $\zeta_B$ we
focus on how to obtain every possible partition of the set of
active terminals, which could be conforming with an almost-optimal solution. By
an appropriate use of $\delta$-nets, similar to \cite{bateni2011approximation},
we are able to ``guess'' (that is, brute-force) a choice of a small number of
net points per active component. Since the number of choices for each point is
at most $|\zeta_B|$ and we choose roughly $O(k^2/\epsilon)$ points in total, the
total number of produced partitions (and hence the running time given by
\autoref{thm:tw-dp}) is of the form $(\log n + k
+\frac{1}{\epsilon})^{O(k^2/\epsilon)}n^{O(1)}$, which is FPT.

Let us now recall a result of \citet{BodlaenderH98} which states that a tree
decomposition of logarithmic height can always be obtained.

\begin{lemma}[\cite{BodlaenderH98}]\label{lem:nice-tree-decomp}
Given a tree decomposition of width $k$ of a graph $G$ on $n$ vertices, there is a
polynomial time algorithm computing a nice tree decomposition of $G$ of
width $O(k$) and height $O(k\log n)$.  \end{lemma}

\begin{proof}%
It is shown in \cite{BodlaenderH98} that any tree decomposition
of width $k$ can be transformed in polynomial time into a tree decomposition of
width $O(k)$ and height $O(\log n)$ where the tree of the decomposition has
maximum degree $3$. It is now not hard to make this decomposition nice by
replacing all nodes with two children by Join nodes, and by inserting between
any node and its parent a sequence of $O(k)$ new nodes so that the symmetric
difference between any node and its parent contains at most one vertex.
\end{proof} 

We also need to reduce the aspect ratio of the given graph to a polynomial.
This can be done using a standard technique, where however we need to make
sure that the treewidth of the given graph remains bounded.  
Note that the aspect ratio of the resulting graph $G'$ in the following lemma
is polynomially bounded in the size of the original graph, 
but not necessarily in the size of $G'$  (because $G'$ may have
significantly fewer vertices).

\begin{lemma}\label{lem:aspect-ratio} %
Given $\eps>0$, an instance of \SF on a
graph $G$ with $n$ vertices, and a (nice) tree decomposition $T$ of width $k$
and height $h$ for $G$, in polynomial time we can compute an instance on a
graph $G'$ with at most $n$ vertices and a (nice) tree decomposition~$T'$ of
width at most $k$ and height $h$ for~$G'$, such that the ratio of the longest
to the shortest edge in~$G'$ is at most $2n/\eps$, and any
$\alpha$-approximation for $G'$ can be converted into an
$(\alpha+\eps)$-approximation for $G$.  
\end{lemma} 

\begin{proof}%
The first
step is to compute a $2$-approximation $F_2$ for \SF in $G$, using the
polynomial time algorithm of~\cite{agrawal1991trees}.  The new graph $G'$ is
obtained from $G$ by first removing all edges of length more than $\cost(F_2)$
and then contracting every edge of length less than
$\frac{\eps}{2n}\cost(F_2)$, where $n$ is the number of vertices of $G$.  If a
contracted edge was incident to a terminal, then the new vertex is declared a
terminal and the demands are updated correspondingly (note that this may
introduce trivial demands from the new terminal to itself if a demand pair is
connected by a path of edges being contracted).  We modify the tree
decomposition $T$ to obtain $T'$ as follows: whenever we contract the endpoints
of an edge $v_1v_2$ into a new vertex $w$, we replace all occurrences of $v_1$
and $v_2$ in $T$ by $w$. It is not hard to see that this keeps a valid decomposition
of the same height and can only decrease the width. Furthermore, if the
original decomposition was nice, the new decomposition can easily be made nice, 
if we contract every bag $B$ with a unique child $B'$
whenever $B=B'$ (which were introduce or forget nodes previously). 
Also, clearly the ratio between longest and shortest edge in
$G'$ is at most $2n/\eps$.

It remains to show that an $\alpha$-approximate solution $F'_\alpha$ in $G'$ is not distorted by much when converting it 
from $G'$ to $G$. Starting with $F'_\alpha$ the conversion is simply done by iteratively uncontracting those edges that were contracted to obtain $G'$ from $G$: 
if the solution becomes infeasible after uncontracting some edge $e$ we just add it to the solution to make it feasible again. 
Let $F$ denote the solution obtained for $G$ from $F'_\alpha$, and note that less than $n$ edges are added to $F'_\alpha$ in this process, as $F$ is a forest in~$G$.
This means that $\cost(F)<\cost(F'_\alpha)+\frac{\eps}{2}\cost(F_2)$ since every contracted edge has length less than $\frac{\eps}{2n}\cost(F_2)$.
Now consider an optimum solution $F^\star$ in $G$. It can be converted into a solution of cost at most $\cost(F^\star)$ in $G'$ by contracting all edges of length less than~$\frac{\eps}{2n}\cost(F_2)$, 
since $F^\star$ cannot contain any of the removed edges of length more than $\cost(F_2)\geq\cost(F^\star)$.
Thus the optimum of $G'$ has cost at most $\cost(F^\star)$, and because~$F'_\alpha$ is an $\alpha$-approximation in $G'$ we get $\cost(F'_\alpha)\leq\alpha\cost(F^\star)$. 
At the same time, $\cost(F_2)\leq 2\cost(F^\star)$, which together with the previous inequality gives 
$\cost(F)<\alpha\cost(F^\star)+2\frac{\eps}{2}\cost(F^\star)=(\alpha+\eps)\cost(F^\star)$, which concludes the proof.
\end{proof}    

For simplicity, in the following we will scale the edge lengths of any given graph so that the shortest edge has length $1$.
In particular, after applying \autoref{lem:aspect-ratio}, the longest edge has length at most~$2n/\eps$.

\subsection{Tree decompositions with bounded height}\label{sec:tw-height}

In this section we informally assume that the height $h$ of the given tree
decomposition is bounded as well as the width $k$. 
Our aim is to prove the following statement, where we restrict ourselves to input graphs of polynomial aspect ratio, which we may do according to \autoref{lem:aspect-ratio} (keeping in mind that $n$ is the number of vertices of the original input graph).

\begin{lemma}\label{lem:near-opt-1} 
Let an instance of \SF on a graph $G$ with at most $n$ vertices 
be given together with a tree decomposition~$T$ of
width~$k$ and height $h$ for $G$. For any $\eps>0$, if the ratio between the
longest and shortest edge of $G$ is at most $2n/\eps$, then there exists a
$(1+\eps)$-approximation $F_\eps$ with the following properties. There exists a
polynomial time algorithm, which for every bag $B$ of~$T$ outputs a partition
$\zeta_B$ of the active terminals~$A_B$, such that each set of~$\zeta_B$
belongs to the same component of $F_\eps$ and
$|\zeta_B|=O(\frac{k^4h^2}{\eps^2}\log\frac{n}{\eps})$.  
\end{lemma}

To prove \autoref{lem:near-opt-1} we first identify the solution $F_\eps$, after which we will show how to compute the partitions $\zeta_B$.

\subsubsection{A near-optimal solution}

The high-level idea to obtain a $(1+\eps)$-approximate solution $F_\eps$ is to connect components of the optimum solution~$F^\star$ that lie very close to each other. In particular, if the distance between two components $C$ and~$C'$ of $F^\star$ is of the form $f(k,h,\eps)\cost(C)$ for some small enough function $f$, then we may hope to add a shortest path between $C$ and $C'$ and charge this additional cost to $C$, in order to obtain a $(1+\eps)$-approximation. Unfortunately, this approach is not viable, since the number of components that are very close to $C$ may be very large, meaning that the function $f$ in the distance bound would have to linearly depend on the number of vertices in order to result in a $(1+\eps)$-approximation. This in turn would mean that the size of the partition $\zeta_B$ would depend polynomially on the number of vertices, making it unsuitable for an FPT time algorithm. This issue lies at the heart of the problem and is the reason for why it is non-trivial to obtain an approximation scheme parameterized by the treewidth. To get around this issue we will measure the distance between components using a modified cost function, which we define next.

Given a bag $B$ of the rooted tree decomposition $T$, we denote by $T_B$
the subtree of $T$ rooted at the node associated with $B$, and by $G_B=G[V_B]$ the graph induced by 
the vertices $V_B$ lying in bags of $T_B$. We also define the graph 
$G^\downarrow_B\subseteq G_B$ as the graph spanned by all edges of $G_B$, 
except those induced by~$B$, i.e., the edge set of $G^\downarrow_B$ is
\[
E(G^\downarrow_B)=\{uv\in E(G_B)\mid u\notin B\lor v\notin B\}.
\] 
The cost of a component $C$ of some \SF solution restricted to $G^\downarrow_B$
only counts the edge weights of $C$ in $G^\downarrow_B$, and is denoted by
\[
\cost^\downarrow_B(C)=\sum_{e\in E(C)\cap E(G^\downarrow_B)}w(e).
\]

Based on these definitions, we fix an optimal solution $F^\star$ and construct
a solution $F_\eps$ by initially setting $F_\eps=F^\star$, and then connecting
components by exhaustively applying the following rule, where we say that two
components $C$ and $C'$ \emph{share} a bag $B$ if $V(C)\cap B\neq\emptyset$ and
$V(C')\cap B\neq\emptyset$:

\begin{description}
    \item[Rule 1:]\customlabel{rule1}{Rule~1} if $C,C'$ are components of $F^\star$ sharing a bag $B$ with $\dist(C,C')\le\frac{\eps}{kh}\cdot\cost^\downarrow_B(C)$ but $C$ and~$C'$ are in different components of $F_\eps$, then add a shortest path of length $\dist(C,C')$ between $C$ and~$C'$ to the solution~$F_\eps$.
\end{description}

\begin{lemma}\label{lem:feps}
The cost of the solution $F_\eps$ obtained by \autoref{rule1} from $F^\star$ is at most $(1+\eps)\cost(F^\star)$.
\end{lemma}

\begin{proof}
It suffices to prove that the cost of all paths added to $F^\star$ in order to obtain $F_\eps$ according to \autoref{rule1}
is at most $\eps\cdot\cost(F^\star)$.
For this we use a charging scheme that charges new paths to components of $F^\star$. 
In particular, we charge a path of length $\dist(C,C')\le\frac{\eps}{kh}\cdot\cost^\downarrow_B(C)$ to component~$C$.

Fix a component $C$ of $F^\star$ and a bag $B$ with $V(C)\cap B \neq \emptyset$. We define $\textrm{charge}(C,B)$ to be the cost we charge to $C$ for operations involving other components of $F^\star$ that share $B$. It is not hard to see that $\textrm{charge}(C,B)\le \frac{\eps}{h}\cdot \cost^\downarrow_B(C)$, because there are at most $k$ other components of $F^\star$ that share $B$. 

For $\ell\in\{0,\ldots,h-1\}$, let $\mathcal{B}_\ell$ be the set of bags of the tree decomposition that appear at distance exactly $\ell$ from the root, i.e., they lie on level $\ell$ of the tree. We now observe that
\[
\sum_{B\in\mathcal{B}_\ell}\textrm{charge}(C,B) \le \sum_{B\in\mathcal{B}_\ell}\frac{\eps}{h}\cdot \cost^\downarrow_B(C) \le \frac{\eps}{h}\cost(C),
\]
where the last inequality follows because if we have two bags $B,B'\in
\mathcal{B}_\ell$, then $E(G^\downarrow_B)\cap E(G^\downarrow_{B'})=\emptyset$:
note that every edge of $E(G^\downarrow_B)$ must be incident on a vertex $v$
that appears in a descendant of $B$, but not in $B$. By the properties of tree
decompositions, notably by the fact that $B$ is a separator of $G$, $v$ cannot
appear in $B'$ or any of its descendants.  Therefore none of its incident edges
are contained in $E(G^\downarrow_{B'})$.  Because
$\sum_{B\in\mathcal{B}_\ell}\cost^\downarrow_B(C)$ is the sum of costs of $C$
over disjoint sets of edges, the sum is a lower bound on the total cost of $C$.

To conclude, we observe that the total charge of $C$ is 
\[
\textrm{charge}(C)\le \sum_{\ell=0}^{h-1}\sum_{B\in\mathcal{B}_\ell}\textrm{charge}(C,B) \le \eps\cost(C).
\]
Therefore, summing over all components of $F^\star$, the total cost of the edges we have added according to \autoref{rule1} is at most $\eps\cdot\cost(F^\star)$.
\end{proof}

\subsubsection{Partitioning active terminals}

We are now ready to prove \autoref{lem:near-opt-1} for the near-optimal solution $F_\eps$ constructed above, for which we will compute the partitions $\zeta_B$ for all bags $B$.
We will use the following two claims for the active terminals~$A_B$ of the given bag $B$.

\begin{claim}\label{claim:1}%
If there exist $t_1,t_2\in A_B$ such that $\dist(t_1,t_2)\le \frac{\eps}{kh}\dist(t_1,B)$, then $t_1,t_2$ are in the same component of $F_\eps$.
\end{claim}

\begin{proof}%
Let $C_1$ be the component that contains $t_1$ in the optimal
solution $F^\star$ from which $F_\eps$ is constructed. We observe that
$\cost^\downarrow_{B}(C_1)\ge \dist(t_1,B)$, because~$C_1$ must contain a path
from $t_1$ to $B$ (as $t_1$ is active) and all the edges of this path are
contained in $E(G_{B}^\downarrow)$. If $t_2$ is contained in $C_2$ in $F^\star$
we therefore have, $\dist(C_1,C_2)\le \dist(t_1,t_2)\le
\frac{\eps}{kh}\dist(t_1,B) \le \frac{\eps}{kh} \cost^\downarrow_{B}(C_1)$.
Since $C_2$ must also intersect $B$, \autoref{rule1} implies that $C_1$ and
$C_2$ are contained in the same component of $F_\eps$.  \end{proof}   

\begin{claim}\label{claim:2} %
    Let $A\subseteq A_B$ and $d\ge 0$ be such that (i) there exists $b\in B$ such that for all $t\in A$ we have $\dist(t,B)=\dist(t,b)$ and $d\le \dist(t,B)\le 2d$ (ii) for all distinct $t,t'\in A$ we have $\dist(t,t')> \frac{\eps}{kh}d$ (iii)~$|A|\ge \frac{8k^2(k+1)h^2}{\eps^2}$. Then, there exists a component of $F_\eps$ that contains all terminals of $A$.
\end{claim}

\begin{proof}%
Consider an active component $C$ for $B$ of the optimum solution $F^\star$. We 
claim that $\cost^\downarrow_{B}(C)\ge |V(C)\cap A|\cdot \frac{\eps}{2kh}d$. 
    To see this, let $C^\downarrow_1,\ldots,C^\downarrow_\ell$ be the components of $C$ when restricting $C$ to $G^\downarrow_B$. Each component $C^\downarrow_i$ is a Steiner tree for the terminals in $V(C^\downarrow_i)\cap A$. Now consider a minimum spanning tree~$U$ on the metric closure derived from $G^\downarrow$ for this vertex set $V(C^\downarrow_i)\cap A$. It is well known that such a minimum spanning tree is a $2(1-\frac{1}{p})$-approximation of the optimum Steiner tree~\cite{KouMB81} on $p$ terminals, and thus $\cost(C^\downarrow_i)\geq\frac{1}{2(1-\frac{1}{p})}\cost(U)$ where $p=|V(C^\downarrow_i)\cap A|$. The distance between any two terminals of $V(C^\downarrow_i)\cap A$ in the given graph $G$ is more than $\frac{\eps}{kh}d$ by property~(ii) of the claim, and because the distance between such terminals can only be more in $G^\downarrow$, every edge of $U$ has cost more than $\frac{\eps}{kh}d$. This means that $\cost(U)\geq (p-1)\cdot\frac{\eps}{kh}d$, and we therefore get $\cost(C^\downarrow_i)\ge \frac{p-1}{2(1-\frac{1}{p})}\cdot \frac{\eps}{kh}d =\frac{p}{2}\cdot\frac{\eps}{kh}d$. Summing over all (vertex disjoint) components~$C^\downarrow_i$ we obtain the claimed inequality $\cost^\downarrow_{B}(C)\ge |V(C)\cap A|\cdot \frac{\eps}{2kh}d$.

    Because each terminal of $A$ belongs to an active component of $F^\star$, of which there are at most~$k+1$, there must exist an active component $C$ with $|V(C)\cap A|\geq\frac{|A|}{k+1}$, which by the above inequality and property~(iii) of the claim gives $\cost^\downarrow_{B}(C)\ge \frac{|A|}{k+1}\cdot \frac{\eps}{2kh}d\ge \frac{kh}{\eps}\cdot 4d$.
    Now note that by property~(i), for all $t,t'\in A$ we have $\dist(t,t')\le 4d$, as we can use a path through $b$. So, if $C'$ is any of the other active components of $F^\star$ also containing a terminal of $A$, we have $\dist(C,C')\le 4d$. We therefore obtain $\dist(C,C')\le \frac{\eps}{kh}\cost^\downarrow_{B}(C)$, and according to \autoref{rule1}, $C$~and~$C'$ are part of the same component of $F_\eps$. In other words, all active terminals of $A$ are in components of $F^\star$ that lie in the one component of~$F_\eps$ containing~$C$.
\end{proof}

Intuitively, \autoref{claim:1} allows us to place terminals of $A$ which are very close to each other into the same set of the partition $\zeta_B$, as placing one terminal in a component forces the placement of the other. Thanks to this claim we can work with an appropriate net. 
If we find a large collection of such net points which also are roughly the same distance from the bag and closest to the same vertex of the bag, \autoref{claim:2} allows us to group them all together in the partition $\zeta_B$.
Armed with these tools, we can now prove the main lemma.

\begin{proof}[Proof of \autoref{lem:near-opt-1}]
To compute the partition $\zeta_B$ in polynomial time, we first partition the active terminals $A_B\cap B$ contained in the bag $B$. For this we simply add a set~$\{t\}$ for each $t\in A_B\cap B$ to $\zeta_B$, which adds at most $|B|\leq k+1$ sets to $\zeta_B$. Let now $A=A_B\setminus B$ be the remaining active terminals.

To partition $A$, let $d=\min_{t\in A_B\setminus B}\dist(t,B)$ and $D=\max_{t\in 
A_B\setminus B}\dist(t,B)$ be the minimum and maximum distances of these active 
terminals from the bag $B$. Then partition $A_B\setminus B$ into $|B|\le k+1$ 
sets $A_1, A_2,\ldots, A_{|B|}$, depending on the vertex of $B$ that is closest 
to each $t\in A$ (breaking ties arbitrarily). That is, for each $A_i$ there 
exists $b\in B$ such that for all $t\in A_i$ we have $\dist(t,B)=\dist(t,b)$. 
Consider now a set $A_i$ and further partition it into 
$r=\lceil\log_2\frac{D}{d}\rceil$ sets $A_{i,0}, A_{i,1},\ldots, A_{i,r-1}$, 
where $A_{i,j}$ contains all $t\in A_i$ such that $\dist(t,B)\in [2^jd, 
2^{j+1}d)$. Now (greedily) compute an $(\frac{\eps}{kh}2^jd)$-net $N_{i,j}$ of 
$A_{i,j}$. We observe that $N_{i,j}$ satisfies the first two conditions of 
\autoref{claim:2} for~$2^jd$, so if $|N_{i,j}|\ge \frac{8k^2(k+1)h^2}{\eps^2}$, 
then we add~$A_{i,j}$ as a set of our partition $\zeta_B$, remove the terminals 
of~$N_{i,j}$ from $A$ and continue the algorithm for the remaining terminals. 
Repeat the previous step for all~$i,j$ for which $N_{i,j}$ is sufficiently 
large. This contributes at most $(k+1)\lceil\log\frac{D}{d}\rceil$ sets to 
$\zeta_B$.

Suppose now that we are left with a set of terminals $A$ such that the procedure above fails to construct a sufficiently large net $N_{i,j}$ to apply \autoref{claim:2}. For every index pair $i,j$, each remaining terminal $t\in A_{i,j}$ is close enough to some net point $t'\in N_{i,j}$ such that we can apply \autoref{claim:1}. We therefore create a set in the partition~$\zeta_B$ for each $t'\in N_{i,j}$, placing into such a set those terminals of $A_{i,j}$ that are closest to~$t'$ (breaking ties arbitrarily). Since we cannot apply \autoref{claim:2} to the remaining sets $A_{i,j}$, each of the at most $(k+1)\lceil\log\frac{D}{d}\rceil$ nets $N_{i,j}$ has size less than $\frac{8k^2(k+1)h^2}{\eps^2}$, which implies $|\zeta_B|\leq O(\frac{k^4h^2}{\eps^2}\log\frac{D}{d})$.

Clearly the above procedure can be implemented in polynomial time, and the fact that every set of~$\zeta_B$ is contained in the same component of $F_\eps$ follows from \autoref{claim:1} and \autoref{claim:2}. Finally, any path in a graph with at most $n$ vertices has less than $n$ edges, so that $\frac{D}{d}<2n^2/\eps$, given that the ratio of the longest to the shortest edge is $2n/\eps$ (note that $d>0$ by definition). Hence the claimed bound of $|\zeta_B|\leq O(\frac{k^4h^2}{\eps^2}\log\frac{n}{\eps})$ follows.
\end{proof}

\subsection{Tree decompositions with logarithmic height}
\label{sec:tw-log-height}

Given a tree decomposition $T$ of logarithmic height, using \autoref{lem:near-opt-1} we are ready to compute a set of partitions $\Pi_B$ of FPT size for each bag $B$, such that a near-optimal solution conforms to $\Pi_B$. In particular, by \autoref{lem:nice-tree-decomp} we may assume that the height of $T$ is $h=O(k\log n)$, which means that the bound on $\zeta_B$ in \autoref{lem:near-opt-1} translates to $O(\frac{k^6}{\eps^2}\log^3\frac{n}{\eps})$.
As in the previous section, we need to apply \autoref{lem:aspect-ratio} in order to bound the aspect ratio of the graph, so that $n$ denotes the number of vertices of the original input graph, while now the graph $G$ has at most $n$ vertices but the ratio between the longest and shortest edge is at most~$2n/\eps$.
We begin by describing how to obtain the near-optimal solution, after which we will identify the partition sets~$\Pi_B$.

\subsubsection{A near-optimal solution}
\label{sec:approx-sol}

\citet{bateni2011approximation} construct a near-optimal solution by modifying the optimum. 
We will use similar techniques to obtain our near-optimal solution, but we construct it by instead modifying the $(1+\eps)$-approx\-imate solution~$F_\eps$ given by \autoref{lem:near-opt-1}.
In particular, we construct a near-optimal $(1+\eps)^2$-approximation~$\widetilde{F}_\eps$ from~$F_\eps$.
The main idea to obtain~$\widetilde{F}_\eps$ is to connect components of~$F_\eps$ if they are very close to one another. 
As before however, doing this naively would incur too much cost for the additional connections.

To make sure that the cost incurred by connecting components of $F_\eps$ is not too large, 
\cite{bateni2011approximation} introduced a partial order on the components based on the structure of a given rooted tree decomposition~$T$.
Let~$C_1,C_2$ be two components of~$F_\eps$ that share a bag $B$ of $T$, i.e.,
$V(C_1)\cap B\neq\emptyset$ and $V(C_2)\cap B\neq\emptyset$.  Since $C_1$ and
$C_2$ are connected subgraphs of the input graph, a basic property of tree
decompositions implies that there are (connected) subtrees $T_1$ and~$T_2$ of
$T$ induced by the respective bags containing vertices of $C_1$ and~$C_2$.
Because these components both contain vertices of~$B$, the node associated with
$B$ is part of both $T_1$ and $T_2$, and therefore the roots of both subtrees
lie on the path from this node to the root of~$T$.  This defines an order on
$C_1$ and $C_2$, and we write $C_1\leq C_2$ if the root of~$T_1$ is farther
from the root of~$T$ than the root of $T_2$ is.  This order is defined for any
two components of~$F_\eps$ that share a bag, and thus we obtain a partial order
on the components of~$F_\eps$, where any components that do not share a bag are
incomparable. 

Using the defined order, \citet{bateni2011approximation} connect components of the optimum solution that are very close to each other.
In order to obtain smaller partition sets, we modify the distance bound used in this procedure compared to \citet{bateni2011approximation}. 
In particular, for any value~$x>0$, let $\pow{x}=2^{\lfloor\log_{2}x\rfloor}$ 
denote the largest power of~$2$ that is at most $x$.
Now, starting with $\widetilde{F}_\eps=F_\eps$ we connect components by exhaustively applying the following rule:
\begin{description}
\item[Rule 2:]\customlabel{rule2}{Rule~2} if $C,C'$ are components of $F_\eps$ with $C\leq C'$ and $\dist(C,C')\leq\frac{\eps}{k}\pow{\cost(C)}$ but $C$ and $C'$ lie in different components of $\widetilde{F}_\eps$, then add a shortest path of length $\dist(C,C')$ between $C$ and $C'$ to the solution $\widetilde{F}_\eps$.
\end{description}

A crucial but subtle observation is that for a component~$C$ of $F_\eps$ there can be many components~$C'\leq C$ at distance at most $\frac{\eps}{k}\pow{\cost(C)}$ to $C$, 
which however are not connected to~$C$ in the resulting solution~$\widetilde{F}_\eps$ according to \autoref{rule2}. 
This makes it non-trivial to find small partition sets~$\Pi_B$.
Contrary to this however, an important property of the order on the components is that for any component~$C$ of $F_\eps$, 
there are at most~$k$ other components $C'$ for which~$C\leq C'$, as we will argue for the following lemma to bound the cost of $\widetilde{F}_\eps$. 
In particular, the lemma implies that $\widetilde{F}_\eps$ is a near-optimal $(1+\eps)^2$-approximation, given that~$F_\eps$ is a $(1+\eps)$-approximation.

\begin{lemma}\label{lem:near-opt-2} %
The cost of the solution $\widetilde{F}_\eps$ obtained by \autoref{rule2} from $F_\eps$ is at most $(1+\eps)\cost(F_\eps)$.
\end{lemma}

\begin{proof}%
Consider a component $C$ of $F_\eps$ and the highest (closest to the root) node of $T$ for which the bag~$B$ contains a vertex of $C$.
Any component~$C'$ with $C\leq C'$ also intersects $B$, and as this bag
has size at most~$k+1$, there can be at most $k$ such components $C'$.
As a consequence, we can charge the additional cost of connecting a component~$C$ with components $C'$ for which $C\leq C'$
to the cost of~$C$. In particular, if $\mathcal{C}$ denotes the set of all components of $F_\eps$, then
the cost incurred by connecting components according to \autoref{rule2} is at most
\[
\sum_{C\in\mathcal{C}}\sum_{C'\in\mathcal{C}: C\leq C'}\frac{\eps}{k}\pow{\cost(C)}\leq
\sum_{C\in\mathcal{C}}\sum_{C'\in\mathcal{C}: C\leq C'}\frac{\eps}{k}\cost(C)\leq
\sum_{C\in\mathcal{C}}\eps\cost(C)=\eps\cost(F_\eps).
\]
Thus adding the cost of connecting components of $F_\eps$ according to \autoref{rule2}, the cost of the resulting solution is $\cost(\widetilde{F}_\eps)\leq\cost(F_\eps)+\eps\cost(F_\eps)=(1+\eps)\cost(F_\eps)$.
\end{proof}    

\subsubsection{Partitioning active terminals}
\label{sec:partitions}

Given the construction of the $(1+\eps)^2$-approximate solution $\widetilde{F}_\eps$ above, the next step is to find a set of partitions $\Pi_B$
of the active terminals $A_B$ for each bag $B$, such that $\widetilde{F}_\eps$ conforms with all sets $\Pi_B$.
In the following, fix a bag $B$ of the given tree decomposition $T$.
The technique used by \citet{bateni2011approximation} is to guess a small net for each active component of bag $B$,\footnote{\citet{bateni2011approximation} refer to these nets as \emph{groups}.}
so that every terminal of $A_B$ close to a net point must be part of the same component in the approximate solution, 
after taking the order on the active components as defined previously into account.
Next we choose a net on the terminals of each active component and bound its size.

\begin{lemma}\label{lem:nets}%
Let $N\subseteq A_B\cap C$ be an $\frac{\eps}{k}\pow{\cost(C)}$-net of the metric induced by 
the active terminals of some active component~$C$.
The size of the net can be bounded by $|N|\leq\lfloor 4k/\eps\rfloor$.\footnote{A slightly worse bound follows from \cite[Lemma~19]{bateni2011approximation}.}
\end{lemma}

\begin{proof}%
Let $U$ be a minimum spanning tree of the metric closure of $N$.
It is well known that a minimum spanning tree is a $2(1-\frac{1}{p})$-approximation to an optimum Steiner tree~\cite{KouMB81} on $p$ terminals, 
and thus we have $\cost(C)\geq\frac{1}{2(1-\frac{1}{|N|})}\cdot\cost(U)$, as $C$ in particular is a Steiner tree for $N$. The distance between any pair of 
net points in~$N$ is more than $\frac{\eps}{k}\pow{\cost(C)}\geq \frac{\eps}{2k}\cost(C)$, 
and given that the spanning tree~$U$ has $|N|-1$ edges we get $\cost(U)>\frac{\eps}{2k}\cost(C)(|N|-1)$. 
Putting these two inequalities together we get $\cost(C)>\frac{\eps(|N|-1)}{4k(1-\frac{1}{|N|})}\cost(C)=\frac{\eps|N|}{4k}\cost(C)$, which implies $|N|\leq\lfloor 4k/\eps\rfloor$ as $|N|$ is an integer.
\end{proof}

Following the algorithm of \citet{bateni2011approximation}, the next step would be to guess such an 
$\frac{\eps}{k}\pow{\cost(C)}$-net for each of the at most $k+1$ active components $C$ of the bag $B$.
By \autoref{lem:nets}, the total number of net points for these at most $k+1$ nets is at most~$\lfloor 4k/\eps\rfloor(k+1)=O(k^2/\eps)$.
Since however there may be up to~$n$ active terminals, guessing these nets for all active components
can result in $n^{O(k^2/\eps)}$ many possible choices, which leads to an XP time algorithm. 
To circumvent this, we instead consider the partition $\zeta_B$ of the active terminals as given by \autoref{lem:near-opt-1}, 
and guess which of the sets of $\zeta_B$ contains a net point.
We will argue that since the size of $\zeta_B$ is~$O(\frac{k^6}{\eps^2}\log^3\frac{n}{\eps})$ there are only $(\frac{k}{\eps}\log \frac{n}{\eps})^{O(k^2/\eps)}$ possibilities, leading to a faster algorithm.

More concretely, to compute a set of partitions $\Pi_B$ that $\widetilde{F}_\eps$ conforms to, 
our algorithm considers every sequence $((S_1,\delta_1),(S_2,\delta_2),\ldots,(S_\ell,\delta_\ell),\rho)$
of at most $k+1$ pairs $(S_j,\delta_j)$ and partitions $\rho$ of the index set~$\{1,\ldots,\ell\}$, 
where each $S_j$ is a subset of the parts of~$\zeta_B$ such that $|S_j|\leq \lfloor 4k/\eps\rfloor$, 
and $\delta_j\in\{2^q\mid q\in\mathbb{N}_0\land 0\leq q\leq \log_2(2n^2/\eps)\}$ is an integer power of~$2$ between $1$
and $2n^2/\eps$, where $n$ is the number of vertices of the original input graph in accordance with \autoref{lem:aspect-ratio}.
From every such sequence, the algorithm attempts to construct a partition of the active terminals, and if it
succeeds adds it to the set~$\Pi_B$. As we will show, in this process the algorithm will successfully construct
one partition $\pi$ of $A_B$ that $\widetilde{F}_\eps$ conforms to.

Before describing how a partition of the active terminals arises from such a sequence, 
we bound the number of these sequences, which determines the running time.
By \autoref{lem:near-opt-1}, $|\zeta_B|=O(\frac{k^6}{\eps^2}\log^3\frac{n}{\eps})$ if the tree decomposition~$T$ has logarithmic height, so that there are
at most $\binom{|\zeta_B|}{\lfloor 4k/\eps\rfloor}=(\frac{k}{\eps}\log\frac{n}{\eps})^{O(k/\eps)}$ possible choices for each~$S_j$. Clearly there are $O(\log\frac{n}{\eps})$ choices for each $\delta_j$, and $\ell^\ell=k^{O(k)}$ possible partitions~$\rho$,
given that $\ell\leq k+1$. Since a sequence contains $\ell$ sets $S_j$, the total number of sequences is bounded by~$(\frac{k}{\eps}\log\frac{n}{\eps})^{O(k^2/\eps)}$.

Each sequence may give rise to a partition $\pi\in\Pi_B$ of the active terminals as follows.
First, let $\pi=\{Y_1,\ldots,Y_{|\rho|}\}$, i.e., $\pi$ has the same number of sets as the partition $\rho$.
Let $U_j=\bigcup_{U\in S_j} U$ denote the set of active terminals in~$S_j$, 
and let $\rho(j)$ be the part of $\rho$ containing $j$.
We distinguish between active terminals $t\in A_B$ that lie in some set $U_j$ and those that do not:
\begin{itemize}
    \item if $t\in U_j$ for some $j\in[\ell]$ then $t\in Y_{\rho(j)}$ (i.e., $U_j\subseteq Y_{\rho(j)}$), and
    \item otherwise, if $p_t\in\{1,\ldots,\ell\}$ denotes the smallest index for which $\dist(t,U_{p_t})\leq\frac{\eps}{k}\delta_{p_t}$, then $t\in Y_{\rho(p_t)}$.
\end{itemize}
If this $\pi$ is a partition of $A_B$ we add $\pi$ to $\Pi_B$, and otherwise we dismiss the current sequence.
Clearly $\pi$ can be constructed in polynomial time, given a sequence.

\begin{lemma}
The $(1+\eps)^2$-approximate solution~$\widetilde{F}_\eps$ conforms to the set $\Pi_B$ of partitions constructed above.
\end{lemma}
\begin{proof}
Consider the $(1+\eps)$-approximate solution $F_\eps$ of \autoref{lem:near-opt-1} from which $\widetilde{F}_\eps$ is constructed 
according to \autoref{rule2}, and the partition $\zeta_B$ of $A_B$ as given by \autoref{lem:near-opt-1}. Let the active components of 
$F_\eps$ be $C_1,\ldots,C_\ell$ indexed according to their order, i.e., $C_j\leq C_{j'}$ if and only if~$j\leq  j'$.
For each active component~$C_j$ we fix an $\frac{\eps}{k}\pow{\cost(C_j)}$-net $N_j$ of size at most~$\lfloor 4k/\eps\rfloor$
according to \autoref{lem:nets}. Now, consider the sequence $((S_1,\delta_1),(S_2,\delta_2),\ldots,(S_\ell,\delta_\ell),\rho)$, where
\begin{itemize}
    \item $S_j$ contains exactly those sets of $\zeta_B$ that contain at least one net point of $N_j$,
    \item $\delta_j=\pow{\cost(C_j)}$, and
    \item $\rho$ is the partition of the index set corresponding to the components of $\widetilde{F}_\eps$, i.e., $\rho(j)=\rho(j')$ if and only if $C_j$ and $C_{j'}$ lie in the same component in $\widetilde{F}_\eps$.
\end{itemize}

Recall that after applying \autoref{lem:aspect-ratio} to the input, the ratio between the shortest and longest edge 
is at most $2n/\eps$, where $n$ is the number of vertices of the original input graph.
Since we assume that the length of the shortest edge is $1$, the cost of any component lies between $1$ and~$2n^2/\eps$, 
given that a component is a tree with less than $n$ edges. Therefore $\pow{\cost(C_j)}\in\{2^q\mid q\in\mathbb{N}_0\land 0\leq q\leq \log_2(2n^2/\eps)\}$, which means that the algorithm will consider the above sequence in some iteration.

We now turn to $\pi=\{Y_1,\ldots,Y_{|\rho|}\}$ constructed for this sequence, and show that it is a partition of~$A_B$ and that $\widetilde{F}_\eps$ conforms to it.
For this, note that no set of $\zeta_B$ contains net points of several active components of $F_\eps$,
since by \autoref{lem:near-opt-1} all active terminals in the same set of $\zeta_B$ also belong to the same component of $F_\eps$.
Thus the sets $S_j$ as defined above (and also the corresponding sets $U_j$) are pairwise disjoint. 
This means that, due to the definition of $\rho$, any two terminals $t\in U_j$ and $t'\in U_{j'}$ end up in the same set
of $\pi$ if and only if $t$ and $t'$ belong to the same component of $\widetilde{F}_\eps$ (as $U_j\subseteq Y_{\rho(j)}$).

Now consider a terminal $t\in A_B$, which does not lie in any $U_j$, and let $q$ be the index of the active 
component $C_q$ of $F_\eps$ containing~$t$. As $\delta_q=\pow{\cost(C_q)}$, $N_q$ is an $\frac{\eps}{k}\delta_q$-net of 
$C_q\cap A_B$. Also, we chose $S_q$ so that $N_q\subseteq U_q$. Hence we get 
$\dist(t,U_q)\leq\dist(t,N_q)\leq\frac{\eps}{k}\delta_q$, and the definition of $p_t$ implies $p_t\leq q$.
Now~$C_{p_t}$ is either equal to~$C_q$, or~$C_q$ is connected to the component~$C_{p_t}$ in the approximate solution~$\widetilde{F}_\eps$
according to \autoref{rule2}: on one hand we have $C_{p_t}\leq C_q$ due to the order of the indices, and at the same time by 
\autoref{lem:near-opt-1} we have $U_{p_t}\subseteq V(C_{p_t})\cap A_B$, which implies
\[
\dist(C_q,C_{p_t})\leq\dist(t,C_{p_t})\leq\dist(t,U_{p_t})\leq\frac{\eps}{k}\delta_{p_t}=\frac{\eps}{k}\pow{\cost(C_{p_t})}.
\]
Hence we can conclude that~$t$ lies in the same component as $C_{p_t}$ in $\widetilde{F}_\eps$. 

In conclusion, adding $U_j$ to $Y_{\rho(j)}$ and~$t$ to $Y_{\rho(p_t)}$ for each terminal $t$ not lying in any $U_j$, 
partitions the terminals according to the components of $\widetilde{F}_\eps$.
Hence $\pi$ is a partition of the active terminals~$A_B$ that is added to $\Pi_B$, and~$\widetilde{F}_\eps$ conforms to it.
\end{proof}

Using all of the above, we can finally prove our main theorem, stating that there is an EPAS for \SF parameterized by the treewidth.

\begin{proof}[Proof of \autoref{thm:tw-EPAS}]

The first steps of our algorithm are to preprocess the given tree decomposition
using \autoref{lem:nice-tree-decomp} so that it is nice and its height is
$O(k\log n)$, and the input graph using \autoref{lem:aspect-ratio} so that the
aspect ratio is bounded (which means that $n$ denotes the number of vertices in
the original input graph). We then compute the partition sets $\Pi_B$ for all
bags $B$ using the above procedure, resulting in partition sets of size
$(\frac{k}{\eps}\log
\frac{n}{\eps})^{O(k^2/\eps)}=2^{O(\frac{k^2}{\eps}\log\frac{k}{\eps})}\cdot
n^{o(1)}$. Here, we are using a well-known Win/Win argument: if
$k^2/\eps<\sqrt{\log n}$, then $(\log n)^{k^2/\eps} = n^{o(1)}$; otherwise,
$\log n\le k^4/\eps^2$, therefore  $(\frac{k}{\eps}\log
\frac{n}{\eps})^{O(k^2/\eps)}=(\frac{k}{\eps})^{O(\frac{k^2}{\eps})}$.

Since each partition of a set $\Pi_B$ can be computed in polynomial time, and
the number of bags of the nice tree decomposition is $O(kn)$, this takes
$2^{O(\frac{k^2}{\eps}\log\frac{k}{\eps})}\cdot n^{O(1)}$ time. Next we apply
\autoref{thm:tw-dp} to compute a solution that is at least as good as
$\widetilde{F}_\eps$ conforming to all $\Pi_B$, in
$2^{O(\frac{k^2}{\eps}\log\frac{k}{\eps})}\cdot n^{O(1)}$ time. Hence we obtain
a $(1+\eps)^2$-approximation~$F$. According to \autoref{lem:aspect-ratio}, $F$
can be converted into a $((1+\eps)^2+\eps)$-approximation to the original input
graph. Since for any $\eps'>0$ we may choose $\eps=\Theta(\eps')$ so that
$((1+\eps)^2+\eps)\leq 1+\eps'$, we obtain an EPAS as claimed.  \end{proof}

\section{Vertex cover}
\label{sec:vc}

In this section we consider the parameterization by the size of a \emph{vertex cover}, which is a set $S\subseteq V$ of vertices such that every edge is incident on at least one of the vertices of $S$. 
We first present an easy FPT algorithm based on the dynamic program given by 
\cref{thm:tw-dp}, and then prove that its runtime dependence on the parameter is 
asymptotically optimal.
\subsection{FPT Algorithm}

Let $S\subseteq V$ be a given vertex cover of size $k$ for the input graph $G$. We will show how to construct a tree decomposition with all required properties of \cref{thm:tw-dp} in order to run the corresponding dynamic program. For this the tree decomposition $(T,\{B_i\}_{i\in V(T)})$ needs to be \emph{nice}.

We may assume w.l.o.g.\ that the vertex cover $S$ contains no terminal: using a standard preprocessing procedure, we can replace any terminal $t\in S$ of the vertex cover by a Steiner vertex $v$ and then connect $t$ with $v$ using an edge of cost~$0$. 
Note that $S$ is still a vertex cover for the preprocessed graph and that the complement set $I=V\setminus S$
of the vertex cover is an independent set containing all terminals. To use \cref{thm:tw-dp},
we will first construct a (trivial) nice tree decomposition for~$I$ and then add $S$ to each bag.

Note that a terminal $t\in R$ can be part of several demand pairs of the \SF instance. 
Consider the \emph{demand graph} $H$ with vertex set $R$ and an edge for each demand pair.
Any subset of $R$ that induces a maximal connected component of $H$ is called a \emph{group}.
Note that every group of terminals must lie in the same connected component of any \SF solution.
For each group $R'\subseteq R$, we create one leaf node of the tree decomposition for each terminal $t\in R'$
and let the corresponding bag contain $t$. We then add a forget node for each such leaf node, which we add
as parent to the leaf with an empty bag. These forget nodes are then connected in a binary tree by adding
join nodes with empty bags (unless the group only contains one terminal in which case we skip this step). 
We proceed in the same way for the Steiner vertices of the independent set $I$,
that is, if we consider $I\setminus R$ to be a group as well we obtain a nice tree decomposition for $I\setminus R$
in which each bag of a leaf node contains one vertex of $I\setminus R$. 
All these trees are then connected using join nodes with empty bags, to obtain a tree decomposition
(of width $0$) for the independent set $I$. Finally, we simply add the vertex cover $S$ to every bag,
which results in a nice tree decomposition (of width $k$) for the graph $G$, such that
every terminal lies in a bag of a leaf node (as~$S\cap R=\emptyset$).

In the obtained tree decomposition, let $V_B$ be the vertices of $G$ contained in all bags in the subtree rooted at the node associated with $B$. By construction, $V_B$ either contains no terminals 
(if $B$ is a bag of the tree decomposition for $I\setminus R$), 
fully contains some groups of $R$ (if $B$ is the bag of the root of a tree decomposition for a group $R'\subseteq R$, 
or if $B$ is a bag of a join node used to connect the tree decompositions for groups in the last step),
or contains some strict subset of only one group of~$R$
(if $B$ is a bag of a non-root node of a tree decomposition for a group~$R'\subseteq R$). 
If no terminals lie in $V_B$ then clearly there are no active terminals for bag $B$. 
However this is also the case if $V_B$ fully contains some groups of~$R$. Hence in both these cases
the set $\Pi_B$ of permutations of active terminals is empty. 
Whenever~$V_B$ contains a strict subset of only one group $R'\subseteq R$,
the active terminals $A_B$ of $B$ are only from this set, i.e., $A_B\subseteq R'$. Thus we can
add the trivial partition $\pi=\{A_B\}$ as the only partition of $\Pi_B$, since all terminals of $R'$ belong to the same component of any solution, including the optimum. 

Clearly, the optimal solution conforms with these sets $\Pi_B$ of permutations, and the total
number $p$ of permutations is at most the number of groups, which is at most $n/2$.
Hence by \autoref{thm:tw-dp} we obtain the algorithm of \autoref{thm:vc}.

\subsection{Runtime lower bound}

Our goal here is to present a reduction showing that the algorithm
we have given for \SF parameterized by vertex cover is essentially optimal,
assuming the ETH. Recall that the ETH is the hypothesis that \textsc{3-SAT} on
instances with $n$ variables cannot be solved in time $2^{o(n)}$. We will give
a reduction that given a \textsc{3-SAT} instance $\phi$, produces an equivalent
\SF instance with vertex cover at most $O(n/\log n)$. We stress that our
reduction works even for unweighted instances.

\begin{theorem} If there exists an algorithm which, given an unweighted \SF
instance on~$n$ vertices with vertex cover $k$, finds an optimal solution in
time $2^{o(k\log k)}n^{O(1)}$, then the ETH is false. \end{theorem}

\begin{proof}

We present a reduction from \textsc{3-SAT}. Before we proceed, we would like to
add to our formula the requirement that the variable set comes partitioned into
three sets in a way that each clause contains at most one variable from each
set. It is not hard to show that this does not affect the complexity of the
instance much, as we demonstrate in the following claim.

\begin{claim}\label{claim:tsat3} %
Suppose that there exists an algorithm that takes as input a
\textsc{3-SAT} instance $\phi$ on $n$ variables and a partition of the
variables into three sets of equal size, such that each clause contains at most
one variable from each set, and decides if $\phi$ is satisfiable in time
$2^{o(n)}$. Then, the ETH is false. \end{claim}

\begin{proof}%
Suppose we start with an arbitrary \textsc{3-SAT} formula $\psi$
on $n$ variables $x_1,\ldots,x_n$. Under the ETH, it should be impossible to
decide if $\psi$ is satisfiable in time $2^{o(n)}$. We will edit $\psi$ to
produce the partition of the variables into three sets. For each variable
$x_i$, we introduce two new variables $x_i', x_i''$, and add to the formula the
clauses $(x_i\to x_i')\land (x_i'\to x_i'')\land (x_i''\to x_i)$.  The
variables of $\psi$ can now be partitioned into three sets $X =
\{x_1,\ldots,x_n\}$, $X'=\{ x_1',\ldots,x_n'\}$, and $X''=\{
x_1'',\ldots,x_n''\}$. Furthermore, because of the clauses we added it is not
hard to see that in any satisfying assignment $x_i, x_i'$, and $x_i''$ must be
given the same value. We then repeat the following: as long as there exists a
clause that contains more than one variable from $X$, arbitrarily pick a
literal of this clause that contains $x_i\in X$ and replace in it $x_i$ by
$x_i'$ or $x_i''$, in a way that the clause contains at most one variable from
each group. The new formula we have constructed in this way is equisatisfiable
to $\psi$, has $n'=O(n)$ variables and $O(n+m)$ clauses, and its variables are
partitioned into three sets so that each clause contains at most one variable
from each set. Therefore, the new formula cannot be solved in time $2^{o(n')}$
under the ETH.  \end{proof}

In the remainder we will then assume that we are given a formula $\phi$ on $3n$
variables which are partitioned into three sets of size $n$ as specified by the
previous claim. Without loss of generality, suppose that $n$ is a power of $4$
(this can be achieved by adding dummy variables). Note that this ensures that
$\frac{\log n}{2}$ and $\sqrt{n}$ are both integers.

We construct an equivalent instance of \SF as follows. Let
$L=\lceil\frac{n}{\log^2n}\rceil$. We  begin by constructing $i$ choice
gadgets, i.e., for $i\in \{1,\ldots,3\log n\}$ we make:

\begin{itemize}

\item $2L$ left vertices, labeled $\ell^i_j$, for $j\in\{0,\ldots,2L-1\}$.

\item $2L$ right vertices, labeled $r^i_j$, for $j\in\{0,\ldots,2L-1\}$.

\item $\sqrt{n}$ middle vertices, labeled $m^i_j$, for
$j\in\{0,\ldots,\sqrt{n}-1\}$. 

\item We connect all middle vertices to all left and right vertices, that is,
for all $j\in\{0,\ldots,2L-1\}$ and $j'\in\{0,\sqrt{n}-1\}$ we connect
$\ell^i_j$ and $r^i_j$ to $m^i_{j'}$.

\item For each $j\in\{0,\ldots,2L-1\}$ we add a demand from $\ell^i_j$ to
$r^i_j$.

\end{itemize}

Notice that the graph we have constructed so far contains $3\log n$ choice
gadgets, each of which has $4L+\sqrt{n}=O(n/\log^2n)$ vertices, so the graph at
the moment contains $O(n/\log n)$ vertices in total.

Before we proceed, let $X=X_a\cup X_b\cup X_c$ be the set of $3n$ variables of
$\phi$ that was given to us partitioned into three sets of size $n$. We
partition $X$ into $3\log n$ groups $X_1,\ldots,X_{3\log n}$ in a way that (i)
$|X_i|\le \lceil n/\log n\rceil$ for all $i\in\{1,\ldots,\log n\}$ and (ii) for
all $i\in\{1,\ldots,\log n\}$ we have $X_i$ is contained in one of
$X_a,X_b,X_c$. This can be done by taking the $n$ variables of $X_a$ and
partitioning them arbitrarily into groups $X_1,\ldots,X_{\log n}$ of size as
equal as possible (therefore at most $\lceil n/\log n\rceil$), and we proceed
similarly for $X_b,X_c$. Rename the variables of $\phi$ so that for each $i$ we
have that $X_i=\{x_{(i,0)},\ldots,x_{(i,\lceil n/\log n \rceil-1)}\}$.

To give some intuition, we will now say that, for $i\in\{1,\ldots, 3\log n\}$,
the choice gadget $i$ represents the variables of the set $X_i$. In particular,
for each $j\in\{0,\ldots 2L-1\}$, we will say that the way that the demand
$\ell^i_j\to r^i_j$ was satisfied encodes the assignment to the $\frac{\log
n}{2}$ variables $\{x_{(i,\frac{j\log n}{2})},\ldots,x_{(i,\frac{(j+1)\log
n}{2}-1)}\}$. More precisely, in our intended solution the demand $\ell^i_j\to
r^i_j$ is satisfied by connecting both terminals to a common middle vertex
$m^i_{j'}$. We can infer the assignment to the $\frac{\log n}{2}$ variables
this represents simply by writing down the binary representation of $j'$, which
is a number between $0$ and $\sqrt{n}-1$, hence a number with $\frac{\log
n}{2}$ bits. Note that this way we represent $2L \cdot \frac{\log n}{2}\ge
\lceil\frac{n}{\log n}\rceil$ variables, that is, we can represent the
assignment to all the variables of the group.

Armed with this intuition, we can now complete our construction. For each
clause $c$ we construct two new vertices, $c_1,c_2$ and add a demand from $c_1$
to $c_2$.  For each literal contained in $c$, suppose that the literal involves
the variable $x_{(i,\frac{j\log n}{2}+\alpha)}$ for $i\in\{1,\ldots,3\log n\}$,
$j\in \{0,\ldots,2L-1\}$, $\alpha\in\{0,\ldots,\frac{\log n}{2}-1\}$. We then
connect $c_1$ to $\ell^i_j$. Furthermore, if $x_{(i,\frac{j\log n}{2}+\alpha)}$
appears positive in $c$, we connect $c_2$ to all $m^i_{j'}$ such that the
binary representation of $j'$ has a $1$ in position $\alpha$. If on the other
hand $x_{(i,\frac{j\log n}{2}+\alpha)}$ appears negative in $c$, we connect
$c_2$ to all $m^i_{j'}$ such that the binary representation of $j'$ has a $0$
in position $\alpha$. In other words, we connect $c_2$ to all the middle
vertices to which $\ell^i_j$ could be connected and are consistent with an
assignment that satisfies $c$ using the current literal. After repeating the
above for all literals of each clause the construction is complete. We set the
target cost to be $B=2m+12L\log n$.

Before we argue about the correctness of the reduction, let us observe that if
the reduction preserves the satisfiability of $\phi$, then we obtain the
theorem, because the instance we constructed has vertex cover $k=O(n/\log n)$
and size polynomial in the size of $\phi$.  Indeed, as we argued the choice
gadgets have $O(n/\log n)$ vertices in total, and all further edges we added
have an endpoint in a choice gadget. If there was an algorithm solving the new
instance in time $k^{o(k)}n^{O(1)}$, this would give a $2^{o(n)}$ algorithm to
decide~$\phi$.

Regarding correctness, let us first observe that if $\phi$ is satisfiable, we
can obtain a valid solution using the intuitive translation from assignments to
choice gadget solutions we gave above. In particular, for each
$i\in\{1,\ldots,3\log n\}$ and $j\in\{0,\ldots,2L-1\}$, we consider the
assignment to variables $\{x_{(i,\frac{j\log n}{2})},\ldots,
x_{(i,\frac{(j+1)\log n}{2}-1)}\}$ as a binary number, which must have a value
$j'$ between $0$ and $\sqrt{n}-1$. We then connect both $\ell^i_j, r^i_j$ to
$m^i_{j'}$. Repeating this satisfies all demands internal to choice gadgets and
uses $3\log n \cdot 4L = 12 L \log n$ edges. Consider now a clause $c$ and the
demand from $c_1$ to $c_2$. Since we started with a satisfying assignment, $c$
must contain a true literal, say involving the variable $x_{(i,\frac{j\log
n}{2}+\alpha)}$. We select the edge from $c_1$ to $\ell^i_{j}$. Furthermore, we
observe that $c_2$ must be a neighbor of all vertices $m^i_{j'}$ such that the
bit in position $\alpha$ of the binary representation of $j'$ agrees with the
value of $x_{(i,\frac{j\log n}{2}+\alpha)}$. Since $\ell^i_{j}$ is already
connected to such a $m^i_{j'}$, we select the edge from that vertex to $c_2$ to
satisfy the demand for this clause. We have therefore spent $2m$ further edges
for the clause demands and have used a budget of exactly $B$.

For the converse direction, suppose we have a solution of cost $B$. We first
observe that each vertex $r^i_j$ must be connected to a middle vertex
$m^i_{j'}$, since all right vertices are terminals, but such vertices only have
edges connecting them to middle vertices. Recall that, for each $i,j$, the left
vertex $\ell^i_j$ must be in the same component of the solution as $r^i_j$,
since there is a demand between these two vertices. Hence, each~$\ell^i_j$ is
in the same component of the solution as some $m^i_{j'}$. We now slightly edit
the solution as follows: suppose there exists a vertex $\ell^i_j$ which is not
directly connected in the solution to any middle vertex~$m^i_{j'}$. Since this
vertex is in the same component as one such vertex $m^i_{j'}$, we add to the
solution the edge connecting them, and since this creates a cycle, remove from
the solution another edge incident on $\ell^i_j$. Doing this repeatedly ensures
that each $\ell^i_j$ is connected to a middle vertex $m^i_{j'}$ in the solution
without increasing the total cost.

We now observe that since each $\ell^i_j$ and each $r^i_j$ is connected to at
least one middle vertex $m^i_{j'}$ in the solution, this already uses a cost of
$3\log n\cdot 4L = 12 L\log n$. Furthermore, for each clause we have
constructed two terminals, each of which must use at least one of its incident
edges, giving an extra cost of~$2m$. Since our budget is exactly $2m+12L\log
n$, we conclude that each terminal constructed for a clause is incident on
exactly one edge, and each $\ell^i_j$ and each $r^i_j$ is connected to exactly
one middle vertex. Crucially, these observations imply the following fact: if
for some $i,j,j'$ we have that $\ell^i_j$ and $m^i_{j'}$ are in the same
component of the solution, then the edge connecting $\ell^i_j$ and $m^i_{j'}$
is part of the solution. To see this, observe that any path connecting
$\ell^i_j$ and $m^i_{j'}$ that is not a direct edge would need to have length at
least $3$. However, no clause terminal can be an internal vertex of such a
path, since clause terminals have degree $1$ in the solution. Furthermore, if
we remove clause terminals from the graph, left and right vertices also have
degree $1$ in the remaining solution, so such vertices also cannot be internal
in the path. Finally, middle vertices are an independent set, so it is
impossible for all internal vertices of a path of length at least $3$ to be
middle vertices.

Armed with the observation that $\ell^i_j$ and $m^i_{j'}$ are in the same
connected component of the solution if and only if they are directly connected,
we are ready to extract a satisfying assignment from the Steiner forest. For
each $i,j$, if $\ell^i_j$ is connected to $m^i_{j'}$ we write $j'$ in binary
and assign to variable $x_{(i,\frac{j\log n}{2}+\alpha)}$, for
$\alpha\in\{0,\ldots,\frac{log n}{2}-1\}$ the value in  position $\alpha$ of
the binary representation of $j'$. We claim that this assignment must be
satisfying. Indeed, consider the clause $c$, and the terminals $c_1,c_2$ which
represent it. Since these terminals have a demand, they must be in the same
component. Because $c_1$ has at most three neighbors which are in different
choice gadgets (as each clause contains variables from distinct groups), we can
see that $c_1$ must be connected to some $\ell^i_j$ and $c_2$ to some~$m^i_{j'}$
in the solution, such that $\ell^i_j$ and~$m^i_{j'}$ are in the same
component, and are therefore directly connected.  But if $\ell^i_j$ is directly
connected to~$m^i_{j'}$ this means that the assignment we extracted from
$\ell^i_j$ gives a value to a variable $x_{(i,\frac{j\log n}{2}+\alpha)}$ which
satisfies the clause $c$, hence we have a satisfying assignment.  \end{proof}

\section{Feedback Edge Set}

A \emph{feedback edge set} of a graph is a set of edges that when
removed renders the graph acyclic. It is well-known that if $G$ is a connected
undirected graph on $n$ vertices and $m$ edges, then all minimal feedback edge
sets of $G$ have size $k=m-n+1$. Indeed, such a set can be constructed in
polynomial time by repeatedly locating a cycle in the graph and selecting an
arbitrary edge of the cycle to insert into the feedback edge set.

In this section we will consider \SF parameterized by the feedback edge set of
the input graph, which we will denote by $k$. Unlike the vertex cover section,
here our main result is positive: we show that \SF can be solved optimally in
time $2^{O(k)}n^{O(1)}$, that is, in time single-exponential in the parameter.
Since we are able to achieve a single-exponential dependence, it is
straightforward to see that this is optimal under the ETH.

\begin{theorem}\label{thm:feseth} %
If there is an algorithm solving \ST in time
$2^{o(k)}n^{O(1)}$, where $k$ is the feedback edge set of the input, then the
ETH is false. \end{theorem}

\begin{proof}%
The proof follows from the sparsification lemma of Impagliazzo,
Paturi, and Zane \cite{ImpagliazzoPZ01} composed with the standard reduction
proving that \ST is NP-complete. We sketch the details. Suppose we are given a
\textsc{3-SAT} formula $\phi$ with $n$ variables and $m$ clauses. The
sparsification lemma shows that in order to disprove the ETH it is sufficient
to show that we can decide if $\phi$ is satisfiable in time $2^{o(n)}$ under
the restriction that $m=\Theta(n)$. We edit $\phi$ to obtain an equisatisfiable
formula $\phi'$ where every variable appears at most $3$ times (for each
variable $x$ appearing $f>3$ times, we replace each occurrence of $x$ with $f$
fresh variables $x_1,\ldots, x_f$ and add the clauses $(x_1\to x_2)\land
(x_2\to x_3)\land\ldots (x_f\to x_1)$). The new formula $\phi'$ has $n'=O(n)$
variables and $m'=O(n)$ clauses. We now execute the chain of reductions showing
that \ST is NP-hard (e.g. from  \cite{karp1975computational}), which produce an
instance on a graph $G=(V,E)$ with $|E|=O(m')$, therefore, $|E|=O(n)$. The new
instance has feedback edge set size $k<|E|$, therefore an algorithm solving the
new instance in time $2^{o(k)}|V|^{O(1)}$ would falsify the ETH.  \end{proof}

Let us now proceed to the detailed presentation of the algorithm. Suppose that
we are given 
a budget~$b$ and we want to decide if there exists a \SF solution $F$ such
that $\cost(F)\le b$.
We start by applying a simple reduction rule.

\begin{description} 
\item[Rule 3:]\customlabel{rule1fes}{Rule 3} Suppose we have a \SF instance on
graph $G$ with weight function $w$ and budget~$b$, such that a vertex $u\in V$
has degree $1$.  If $u\not\in R$, then delete $u$.  If $u\in R$, let $v$ be the
unique neighbor of $u$. Then  set $b':=b-w(uv)$, delete $u$ from the graph and
the demand $\{u,v\}$ from $D$ if it exists, and replace, for each $x\in
V\setminus\{u,v\}$ such that $\{u,x\}\in D$ the demand $\{u,x\}$ with the
demand $\{v,x\}$.
\end{description}

\begin{lemma}\label{lem:leaf} %
\ref{rule1fes} is safe. 
\end{lemma}

\begin{proof}%
If $u\not\in R$ then no optimal solution contains the edge $uv$,
so it is safe to delete $u$. If $u\in R$ then all feasible solutions contain
the edge $uv$. \end{proof}

Observe that if we apply \ref{rule1fes} exhaustively, then the minimum degree of the graph is~$2$. As we show next,
relatively few vertices can have higher degree.

\begin{lemma}\label{lem:deg3} %
Suppose we have a \SF instance with feedback edge
set of size $k$ and minimum degree at least $2$.  Then $G$ contains at most
$2k$ vertices of degree at least $3$. 
\end{lemma}

\begin{proof}%
We observe that if our graph has a feedback edge set of size
$k$, then $m= k+n-c$, where $c$ is the number of connected components of $G$.
This implies that $\sum_{v\in V}d(v) = 2m = 2k+2n-2c$. Let $V_2$ be the set of
vertices of degree exactly $2$ and $V_3=V\setminus V_2$ be the set of vertices
of degree at least $3$. We have $\sum_{v\in V}d(v) \ge 2|V_2|+3|V_3| = 2n
+|V_3|$. We conclude that $|V_3|\le 2k-2c$.  \end{proof}

In the remainder we will assume that we have a \SF instance $G=(V,E)$ with a
feedback edge set $H\subseteq E$ of size $k$, to which \ref{rule1fes} can no longer be
applied. We will say that a vertex $v$ is \emph{special} if $v$ is incident on
an edge of $H$ or $v$ has degree at least $3$. By \cref{lem:deg3} we know that
$G$ contains at most $4k$ special vertices.

We define a \emph{topological edge} (topo-edge for short) as follows: a path
$P$ in $G$ is a topological edge if the two endpoints of $P$ are special
vertices and all internal vertices of $P$ are non-special. Note that by this
definition, all edges of $H$ form topo-edges, since the endpoints of such edges
are special.  We observe the following:

\begin{lemma}\label{lem:topo} %
Suppose we have a graph $G$ with
feedback edge set of size $k$ and minimum degree at least~$2$. Then $G$
contains at most $5k$ topological edges.  
\end{lemma}

\begin{proof}%
Let $V_s$ be the set of special
vertices and $V_t=V\setminus V_s$. If we have more than $5k$ topological edges
in $G$ then $\sum_{v\in V_s}d(v)\ge 10k$. This is because each topological edge
contributes at least $2$ in the sum $\sum_{v\in V_s}d(v)$.  On the other hand,
if $c$ is the number of connected components of $G$, we have $2n+2k-2c = 2m =
\sum_{v\in V}d(v) = \sum_{v\in V_s}d(v) + \sum_{v\in V_t} d(v) = \sum_{v\in
V_s} d(v) +2|V_t|$.  However, $|V_t|\ge n-4k$ by \cref{lem:deg3} and the fact
that at most $2k$ vertices are incident on $H$.  Hence, $2n+2k-2c\ge \sum_{v\in
V_s}d(v) +2n-8k$.  This implies that $\sum_{v\in V_s}d(v)<10k$.  Hence, it is
impossible to have more than $5k$ topological edges. \end{proof}

We are now ready to state the main algorithmic result of this section.

\begin{theorem} 
There is an algorithm that solves \SF on instances with $n$
vertices and a feedback edge set of size $k$ in $2^{O(k)}n^{O(1)}$ time. 
\end{theorem}

\begin{proof}
Call the set of special vertices $V_s$ and let $V_t=V\setminus V_s$. For the rest of
this proof and for the sake of the analysis, fix an optimal solution $F^\star$.

To begin, we guess which of the $5k$ topological edges according to \cref{lem:topo} are fully used in the
optimal solution. To be more precise, we will say that a topo-edge $P$ is fully
used in $F^\star$ if all edges of the path $P$ are contained in~$F^\star$. This gives
$2^{5k}$ possibilities. In the remainder we will assume that we have correctly
guessed the set of topo-edges which are fully used in $F^\star$.

We now observe that for any two vertices $u,v\in V_s$ we have enough
information to deduce whether~$u,v$ are in the same connected component of
$F^\star$. More precisely, we construct an auxiliary graph $G_s$ with vertex
set $V_s$ that contains an edge between two vertices $u,v\in V_s$ if there
exists a fully used topo-edge whose endpoints are $u,v$.  We now claim that two
vertices $u,v\in V_s$ are in the same component of $F^\star$ if and only if
$u,v$ are in the same connected component of $G_s$. Indeed, if two vertices
$u,v$ are in the same component of $G_s$, then clearly there is a path
connecting them in $F^\star$ going through fully used topo-edges; conversely if
$u,v$ are in the same component of $F^\star$ and the path connecting them goes
through the special vertices $u_1=u, u_2,\ldots, u_\ell=v$ (and all other
vertices are non-special), then the path $u_1,\ldots, u_\ell$ also exists in
$G_s$, as the topo-edge connecting $u_i, u_{i+1}$ must be fully used.

Because of the above, we can now assume that we have a partition $\rho$ of $V_s$
such that $u,v$ are in the same set of $\rho$ if and only if $u,v$ are in the
same connected component of $F^\star$. Notice that this implies that we can remove
from the instance all demands $\{u,v\}\in D$ such that $u,v\in V_s$: if $u,v$ are in
the same set of $\rho$ the demand is automatically satisfied by our guess of the
fully used topo-edges; while if $u,v$ are in distinct sets of $\rho$, we know
that our guess is incorrect and we reject the current instance. Every remaining
demand of our instance is therefore incident on at least one non-special
vertex.

What remains is to decide for topo-edges which are not fully used, which of
their incident edges belong in $F^\star$. Note that this is trivial for topo-edges
consisting only of a single edge, since fully using such a topo-edge is
equivalent to placing the corresponding edge in the solution. We therefore
focus on topo-edges which contain at least one internal (non-special) vertex. 

For this we proceed in several steps.  First, suppose we have a non-fully-used
topo-edge $P$ whose endpoints are adjacent to $u,v\in V_s$ such that $u,v$ are
in the same component of $F^\star$. We edit the instance so that demands with one
endpoint in the interior of $P$ also have their other endpoint in $P$. More
precisely, for each demand $\{x,y\}\in D$ such that $x$ is an internal vertex
of $P$ and $y\not\in P$, we remove $\{x,y\}$ from $D$ and replace it with the
demands $\{x,u\}$ and $\{y,u\}$. It is not hard to see that this is safe,
because any path satisfying the demand $\{x,y\}$ would have to go either
through $u$ or through $v$, but $u,v$ are in the same component of $F^\star$ thanks
to other, fully-used topo-edges, so routing the demand through $u$ or $v$ is
the same.

Consider then a topo-edge $P$ whose endpoints $u,v$ are in the same component
of $F^\star$ and where all demands with one endpoint in an internal vertex of $P$
have the other endpoint in $P$. We simplify the instance by branching: select
an edge $e\in P$, delete $e$ from the instance, and apply \ref{rule1fes} exhaustively
on internal vertices of $P$, until all such vertices are removed.  Since we
have guessed that $P$ is not fully used, at least one of the instances we
produced is equivalent to the original, that is, at least one choice of edge to
delete indeed deletes an edge not used by the optimal solution. It may seem
that since we are branching on $n$ possibilities, this branching will lead to a
running time of $n^k$.  However, we observe that after removing any edge of $P$
and exhaustively applying \ref{rule1fes} we obtain instances which have (i)
the same graph, as all internal vertices of $P$ have been deleted and all other
vertices are unchanged (ii) the same set of demands, as all demands with one
endpoint in an internal vertex of $P$ have either been removed or replaced with
the demand $\{u,v\}$ (which is satisfied by the fully used topo-edges, so can
be removed) and other demands are unchanged.  Hence, among the at most $n$
instances this branching produces, it suffices to select the one with the
maximum remaining budget and solve that, to decide if the original is a Yes
instance. In other words, the branching procedure of this paragraph is a
polynomial-time reduction rule which allows us to eliminate all topo-edges
whose endpoints are in the same component of~$F^\star$.

In the remainder we thus assume that every topo-edge $P$ that is not fully used
has endpoints $u,v\in V_s$ which are in distinct components of $F^\star$. Next, we
deal with the case of ``internal'' demands. Suppose that there exists a
topo-edge $P$ with endpoints $u,v\in V_s$ that contains an internal demand,
that is, there exist $w_1,w_2\in P\setminus\{u,v\}$ such that $\{w_1,w_2\}\in
D$.  Then, all edges in the path from $w_1$ to $w_2$ in $P$ must belong in
$F^\star$, because every other solution that connects $w_1$ to $w_2$ would put
$u,v$ in the same component. We can therefore contract all the edges of the
path from $w_1$ to $w_2$ and adjust our budget and our demands
accordingly: we decrease our budget by the total cost of the edges of
the path from $w_1$ to~$w_2$, we remove all demands that have both endpoints in
that path, and for demands that have one endpoint in that path, we replace that
endpoint by the vertex that results from the contraction of the path.

We now arrive at the case where the endpoints of each topo-edge are adjacent to
vertices from distinct components of $F^\star$ and demands with one endpoint in the
interior of a topo-edge have the other endpoint outside of the topo-edge or in
$V_s$.

We distinguish several cases:

\begin{enumerate}

\item There exists a topo-edge $P$ adjacent to $u,v\in V_s$, an internal vertex
$w\in P\setminus\{u,v\}$ and a vertex $w'\in V_s$ such that $\{w,w'\}\in D$. If
$w'$ is in the same component of $F^\star$ as $u$ (respectively $v$), we include in
the solution all edges in the path in $P$ from $w$ to $u$ (respectively $v$),
contract the selected edges and update our budget and demands accordingly, as above. If
$w',u,v$ are in distinct components of~$F^\star$, then we conclude that the current
guess is incorrect and reject the instance. Correctness of these actions
follows if we assume that the partition $\rho$ of $V_s$ we have computed
corresponds to the connected components of $F^\star$, because in the latter case
any solution that connects $w$ to $w'$ will place $w'$ in the same component as
one of $u,v$, and in the former case we are forced to use the selected path, as
otherwise $u,v$ would end up in the same component of $F^\star$.

\item There exist two topo-edges $P_1,P_2$ adjacent to $u_1,v_1\in V_s$ and
$u_2,v_2\in V_s$ respectively, and vertices $w_1\in P_1$ and $w_2\in P_2$ such
that $\{w_1,w_2\}\in D$. If $u_1,v_1,u_2,v_2$ are in four distinct components
of~$F^\star$ we reject the current guess, as it is impossible to place $w_1,w_2$ in
the same component without also placing some of $u_1,v_1,u_2,v_2$ in the same
component.

\item If $u_1,u_2,v_1,v_2,w_1,w_2$ are as previously but $u_1,v_1,u_2,v_2$ are
in three distinct components of $F^\star$, we can assume without loss of
generality that $u_1,u_2$ are in the same component. We replace the demand
$\{w_1,w_2\}$ with the demands $\{w_1,u_1\}$ and $\{w_2,u_2\}$ and reduce to a
previous case.

\end{enumerate}

Finally, if none of the previous cases apply we have arrived at an instance
where all remaining demands $\{w_1,w_2\}\in D$ satisfy that $w_1,w_2$
belong in two distinct topo-edges $P_1,P_2$, which are incident on $u_1,v_1\in
V_s$ and $u_2,v_2\in V_s$ respectively, such that $u_1,u_2$ are in the same
component of $F^\star$, and so are~$v_1,v_2$, but the component of $u_1,u_2$ is
distinct from the component of $v_1,v_2$.  We will find the best way to satisfy
such demands by solving an auxiliary problem.

Fix two sets $C_1,C_2$ of the partition $\rho$ of $V_s$ which we have computed
and consider every topo-edge $P$ with one endpoint in $C_1$ and the other in
$C_2$. We construct a new instance of \SF on a graph $G_2$ by taking the union
of all such topo-edges and then contracting all vertices of $C_1$ into a single
vertex $c_1$ and all vertices of $C_2$ into a single vertex $c_2$. We include
in the new instance all demands with at least one endpoint on one of the
internal vertices of the topo-edges we used; note that such demands also have
the second endpoint in $G_2$. Let $G_1$ be the instance induced from the
original graph if we delete all internal topo-edge vertices which appear in
$G_2$. Note that every demand of the original instance appears in either $G_1$
or $G_2$.

We will now state two claims:

\begin{claim}\label{claim:g1g2} %
If the optimal \SF solution on the
instance $G_2$ constructed above has cost~$b_2$, then we have the following:
$G$ has a solution of cost at most $b$ consistent with the guess $\rho$ if and
only if $G_1$ has a solution consistent with the guess $\rho$ of cost at most
$b-b_2$.  \end{claim}

\begin{claim}\label{claim:g2opt} %
The optimal solution to $G_2$ can be
computed in polynomial time by a reduction to the \textsc{Min Cut} problem.
\end{claim}

Let us explain why the claims are sufficient to conclude our algorithm. We
consider every pair of sets~$C_1,C_2\in\rho$ (of which there are $O(k^2)$) and
for each such pair the claims imply that we can decompose the instance into two
instances $G_1,G_2$, such that $G_2$ can be solved in polynomial time, and
using the optimal value we calculate for $G_2$ we can reduce solving $G$ to
solving $G_1$. Repeating this for all  pairs results in an instance with no
demands. Putting everything together, we have that for one of $2^{5k}$ possible
guesses (on which topo-edges are fully used) we apply a series of
polynomial-time reduction rules that allow us to decompose the instance into
$O(k^2)$ polynomial-time solvable sub-problems. We therefore obtain an exact
algorithm running in $2^{O(k)}n^{O(1)}$ time. \qedhere

\begin{proof}[Proof of \cref{claim:g1g2}]

If $G_1$ has a solution of cost $b-b_2$ consistent with $\rho$, then we can form
a solution for $G$ by taking the union of the solution for $G_1$ with an
optimal solution for $G_2$.  This will have cost at most~$b$. Furthermore,
recall that all demands of $G$ appear in either $G_1$ or $G_2$. Demands that
appear in $G_1$ are clearly satisfied by the new solution in $G$, while demands
that appear in $G_2$ are satisfied because the solution in $G_1$ is consistent
with $\rho$, so it contains paths between any two $u,v\in C_1$ for each $C_1\in
\rho$.

For the converse direction, suppose $G$ has a solution of cost $b$ consistent
with $\rho$. We observe that this solution restricted to $G_2$ is a feasible
solution (which furthermore places $c_1,c_2$ in distinct components), hence
must have cost at least $b_2$. Therefore, the solution restricted to edges of
$G_1$ has cost at most $b-b_2$. Because all topo-edges included in $G_2$ are
topo-edges which are not fully used (according to the guess that gave us
$\rho$), the solution we construct in $G_1$ is still consistent with $\rho$ and
satisfies all demands.  \end{proof}
\begin{proof}[Proof of \cref{claim:g2opt}]

Before we begin, we perform a basic simplification step. If the instance
contains a Steiner vertex $v$ of degree $2$ (that is, a vertex not incident on
any demand), with neighbors $u_1,u_2$, then we remove $v$ from the instance and
add an edge $u_1u_2$ with weight equal to $w(vu_1)+w(vu_2)$. It is not hard to
see that the new instance is equivalent ($v$ would only be used in a solution
if both its incident edges are used), and we now know that all vertices of
degree $2$ are terminals.

Recall that we have a graph $G_2$ with two special vertices $c_1,c_2$ such that
the graph consists of a collection of parallel paths with endpoints $c_1,c_2$,
and furthermore every demand is between two internal vertices of distinct
paths. For the purposes of the larger algorithm we are interested in computing
the best solution where $c_1,c_2$ are in distinct components, but for the sake
of completeness let us briefly note that $G_2$ can be solved to optimality
without this constraint, as the best solution where $c_1,c_2$ are in the same
component is just a minimum cost spanning tree of $G_2$ (here we are using the
fact that all internal vertices of all paths are terminals).

In order to compute the best solution that places $c_1,c_2$ into distinct
components, we will reduce the problem to \textsc{Min Cut}. Let $N$ be a
sufficiently large value, for example set $N$ to be the sum of all edge weights
of the instance. We construct a \textsc{Min Cut} instance on the same graph but
with weight function $w'(e)=N-w(e)$. Furthermore, for all $w_1,w_2$ such that
$\{w_1,w_2\}\in D$ we add an edge $w_1w_2$ and set $w'(w_1w_2)=n^2N$.

Our claim is now that if $F_c$ is a set of edges that gives a minimum weight
$c_1-c_2$ cut in the new instance, then the complement of $F_c$ is a minimum
cost \SF solution for $G_2$ that places $c_1,c_2$ in distinct components.

To prove the claim, suppose that $F_c$ is a minimum-weight $c_1-c_2$ cut in the
new graph. We observe that $F_c$ cannot include any of the edges we added
between the endpoints of demands $(w_1,w_2)\in D$, as such edges have a very
high cost (deleting every other edge would be cheaper). Furthermore, because
all edges have positive weight and the cut $F_c$ is minimal, removing $F_c$
from the graph must leave exactly two connected components, one containing each
of $c_1,c_2$. Hence, for each $(w_1,w_2)\in D$, if we keep in the graph all
edges not in $F_c$, $w_1,w_2$ are in the same component, and we have a feasible
\SF solution for all the demands. In the other direction, consider an optimal
\SF solution that places $c_1,c_2$ in distinct components, and let $F_c'$ be
the set of edges of $G_2$ \emph{not} included in the solution. $F_c'$ must be a
valid $c_1-c_2$ cut, because it contains at least one edge from each
topological edge connecting $c_1$ to $c_2$ (otherwise $c_1,c_2$ would be in the
same component); and as each demand $(w_1,w_2)\in D$ is satisfied, therefore,
$w_1,w_2$ are either in the component of $c_1$ or in the component of $c_2$. We
have therefore established a one-to-one mapping between optimal minimum cuts
and optimal \SF solutions and conclude the claim by observing that by
minimizing the weight of $F_c$ in the \textsc{Min Cut} instance, we are
maximizing the weight of non-selected edges in the \SF instance (thanks to the
modified weight function), hence we are selecting an optimal \SF solution.
\end{proof}

\end{proof}

\printbibliography

\newpage
\appendix

\end{document}